\documentclass[acmsmall,nonacm]{acmart}

\AtBeginDocument{%
  }

\usepackage[english]{babel}
\usepackage{amsthm}
\usepackage{braket}
\usepackage{algorithm}
\usepackage[noend]{algpseudocode}
\usepackage{graphicx}
\usepackage{wrapfig}

\newcommand{\cev}[1]{\reflectbox{\ensuremath{\vec{\reflectbox{\ensuremath{#1}}}}}}

\theoremstyle{definition}
\newtheorem{definition}{Definition}[section]

\theoremstyle{plain}
\newtheorem{theorem}{Theorem}[section]

\theoremstyle{plain}
\newtheorem{example}{Example}[section]

\theoremstyle{remark}

\begin{document}

\title{Space-time Peer-to-Peer Distribution of Multi-party Entanglement for Any Quantum Network}

\author{Yuexun Huang}
\email{yesunhuang@uchicago.edu}
\affiliation{%
  \institution{Pritzker School of Molecular Engineering, University of Chicago}
  \city{Chicago}
  \state{Illinois}
  \country{USA}
}

\author{Xiangyu Ren}
\email{xiangyu.ren@ed.ac.uk}
\affiliation{%
  \institution{School of Informatics, University of Edinburgh}
  \city{Edinburgh}
  \country{United Kingdom}
}

\author{Bikun Li}
\affiliation{%
  \institution{Pritzker School of Molecular Engineering, University of Chicago}
  \city{Chicago}
  \state{Illinois}
  \country{USA}
}

\author{Yat Wong}
\affiliation{%
  \institution{Pritzker School of Molecular Engineering, University of Chicago}
  \city{Chicago}
  \state{Illinois}
  \country{USA}
}

\author{Zhiding Liang}
\affiliation{
    \institution{Rensselaer Polytechnic Institute}
    \city{Troy}
    \state{New York}
    \country{USA}
}

\author{Liang Jiang}
\email{liang.jiang@uchicago.edu}
\affiliation{%
  \institution{Pritzker School of Molecular Engineering, University of Chicago}
  \city{Chicago}
  \state{Illinois}
  \country{USA}
}

\renewcommand{\shortauthors}{Huang et al.}

\begin{abstract}
Graph states are a class of important multiparty entangled states, of which Bell pairs are the special case. Realizing a robust and fast distribution of arbitrary graph states in the downstream layer of the quantum network can be essential for further large-scale quantum networks. We propose a novel quantum network protocol called P2PGSD inspired by the classical Peer-to-Peer (P2P) network to efficiently implement the general graph state distribution in the network layer, which demonstrates advantages in resource efficiency and scalability over existing methods for sparse graph states. An explicit mathematical model for a general graph state distribution problem has also been constructed, above which the intractability for a wide class of resource minimization problems is proved and the optimality of the existing algorithms is discussed. In addition, we leverage the spacetime quantum network inspired by the symmetry from relativity for memory management in network problems and use it to improve our proposed algorithm. The advantages of our protocols are confirmed by numerical simulations showing an improvement of up to 50\% for general sparse graph states, paving the way for a resource-efficient multiparty entanglement distribution across any network topology.
\end{abstract}






\maketitle

\section{INTRODUCTION}
The quantum Internet~\cite{cacciapuoti2019quantum,wehner2018quantum,kimble2008quantum} has long been expected to be the game changer in information science, having a profound impact on many fields beyond its classical counterpart, including communication, computation, and metrology. A reliable quantum network shall be fundamental to realizing fruitful applications in quantum communication protocols such as multiparty quantum key distribution~\cite{bennett2014quantum}, quantum secret sharing~\cite{cleve1999share}, and quantum voting~\cite{buhrman2022quantum,vaccaro2007quantum}. In addition, cloud-based computation seems more likely to be the common practice for the public to access quantum speed-up due to the complicated and expensive maintenance of quantum computers~\cite{ravi2021quantum,devitt2016performing}, not to mention it has been a common practice for industry in the classical world. Some protocols such as quantum delegated computing~\cite{barz2012demonstration,fitzsimons2017private}, and quantum private queries~\cite{giovannetti2008quantum} even demonstrate unique security advantages in quantum worlds that lead to the idea of quantum data centers~\cite{liu2023data}. Finally, entanglement-boosted ultrahigh precision metrology, such as quantum clock synchronization~\cite{komar2014quantum} and quantum telescopes~\cite{gottesman2012longer} also demand the implementation of reliable quantum networks due to the need to combine physically separated signals.

The structure of the quantum networks can be made stack-like just like the classical one~\cite{pirker2019quantum}, where the aforementioned applications mostly lie on the application layer built upon the availability of reliable downstream layers responsible for information exchanges and creating reliable entanglement among operating clients, the latter of which is the unique property of the quantum networks. There has been tremendous work focusing on routing the entanglement to create Bell pairs efficiently between requesting clients~\cite{shi2020concurrent,caleffi2017optimal,chakraborty2019distributed,pant2019routing}, after which multiparty entanglement can be created using established Bell pairs and local operations. However, this might require additional memory time, which is expensive in the quantum case. On the other hand, the graph states, of which Bell pairs can be considered a special case, providing intrinsic multiparty entanglement, can be considered as the more general resources for quantum networks~\cite{hein2006entanglement}. It is important to generalize the routing of the entanglement in the network layer from Bell pairs to general graph states directly, which is a distinct problem from classical networks arising from the quantum property of multiparty entanglement and could provide advantages in resource usage compared to the previous scheme, which has been shown for quantum cryptography~\cite{memmen2023advantage}.

Although some pieces of literature have researched this problem\cite{cuquet2012growth,pirker2018modular,meignant2019distributing,fischer2021distributing,xie2021graph,fan2024optimized}, most of them just care about bell-pair usages without considering the stochastic nature of the quantum network, where the quantum connections between nodes could fail with relatively high probability. In addition, the local operation performed at each node can also transform and stitch the separated parts of the graph state together, which can lead to optimized performance by a carefully designed protocol. In particular, we consider the general fusion operation here and define a general class of graph state distribution problems. Based on this problem, we analyze its intractability and hardness, which hinders any efficient and exact algorithms for an optimal solution. Then we leverage the idea of Peer-to-Peer distribution~\cite{barkai2001peer} from the classical network to construct an efficient approximate algorithm called Peer-to-Peer Graph State Distribution (P2PGSD) for our graph state distribution protocols lying in the network layer, which is further improved to Spacetime Peer-to-Peer Graph State Distribution (ST-P2PGSD) using spacetime symmetry. We also adapt the previously proposed center-distributing algorithms Graph State Transfer (GST)~\cite{fischer2021distributing} for the quantum network model by proving an important and non-trivial graph theorem. A detailed comparison between the adopted algorithm named Modified Graph State Transfer (MGST) and the new algorithms is performed to demonstrate their advantageous region.  To our knowledge, this is the first work to comprehensively study the distribution of any graph states in any network topology with a well-defined mathematic model, concrete protocols, and metrics incorporating the random nature of a quantum network.

Our contributions are listed below.
\begin{enumerate}
    \item \textbf{Rigorous Mathematical Model.} We develop a concrete mathematical framework for the problem of distributing any graph state across any network topology using fusion operation, upon which the intractability under various conditions and optimization goals is proven.
    \item \textbf{New Efficient Distribution Algorithms.}  Two new decentralized algorithms are developed to solve the above problem called P2PGSD and ST-P2PGSD based on the classical idea of Peer-to-Peer transfer, the performance of which is demonstrated to have great advantages when the graph is sparse compared to the updated centralized algorithms MGST. We also improve and develop new recovery algorithms for basic Bell pair routing.
    \item \textbf{Powerful Network Memory Management Method.} The tool named the space-time quantum network, which utilizes the symmetry between space and time, is proposed to analyze memory consumption in distributed quantum systems. We show its strong versatility and convenience in analyzing memory consumption by employing it to design algorithms and prove the intractability of the memory optimization problem.
\end{enumerate}

The remainder of this paper is organized as follows. The basic principles of graph state manipulation, the mathematical model of the network, and the main theorems are presented in Section \ref{sec:backgroundAMoti}, after which the main algorithms and protocols designs are introduced in Section \ref{sec:algorithms}. Then a numerical simulation is performed, the results of which are analyzed in Section \ref{sec:evalution}. We also discuss the immediate application by offering a new framework for distributing any quantum computation as well as a concrete example in Section \ref{sec:applitions}. Section \ref{sec:relate} provides the summary for the comparison between our work and other work concerning entanglement routing, and we conclude this work in Section \ref{sec:conclusion}.

\section{BACKGROUND AND MOTIVATION}
\label{sec:backgroundAMoti}
\subsection{Graph State}
A quantum graph state $\ket{G}$ is a specific type of multipartite entanglement state that is intrinsically associated with a graph $G=(V,E)$, where $V$ is the set of vertices and $E$ is the set of edges connecting these vertices. In graph $G$, a vertex $v$ represents a physical qubit, whereas an edge $e$ indicates the entanglement between the two qubits.
The formal definition of a graph state $\ket{G}$ is given by:
\begin{equation}
    \ket{G} = \prod_{(i,j)\in E}{CZ_{(i,j)}\ket{+}^{\otimes  V }}
\end{equation}
Here, $\ket{+}^{\otimes  V }$ denotes the tensor product of $\lvert V \rvert$ copies of $\ket{+}$ state at all the vertices of $G$, and $CZ_{(i,j)}$ is the controlled-Z gate operation applied to the pair of qubits corresponding to the edge $(i,j) \in E$. 
Typically, the generation of a graph state is achieved through two steps:
1) Prepare a set of qubits in initial states, corresponding to all vertices in $V$; 
2) For each pair of qubits where the entanglement is desired, apply a CZ gate to create interaction. 

\subsection{Principles of P2P Graph State Distribution}\label{sec:principle}
In this paper, we will provide a detailed definition of the problem under consideration. In the remainder of this paper, we will use nodes and channels to describe the vertices and edges of the network topology, while the vertices and edges of the graph of the graph state. To begin with, the network model under consideration is defined by a tuple as
\begin{equation}
\label{eq:definitionOfNetwork}
\mathscr{N}\triangleq(G_N,W_c,W_m).
\end{equation}
Here, $G_N=(V_N,E_N)$ is an undirected graph representing the topology of the network with $V_N$ denoting the set of nodes and $E_N$ denoting the set of channels. $W_c:E_N\rightarrow\mathbb{N}^+$ is a map that represents the width of the channel such that each channel $c\in E_N$ can establish $W_c(c)$ Bell pairs per cycle time.  And $W_m:V_N\rightarrow\mathbb{N}$ is a map that represents the long-term memory available for each node. 

The graph state to be distributed is simply represented by a graph $$G_S=(V_S,E_S)$$ as defined above with $V_S$ being the set of vertices and $E_S$ the set of edges. And for distributing a graph state $G_S$ across a network $\mathscr{N}$, we want certain nodes of the network to hold a certain set of vertices (qubits) locally of the graph state at the end of the distribution protocol. This requirement is captured by a map $$\alpha_{D}:V_S\rightarrow V_N,$$ that assigns the vertex $v_s\in V_S$ to the network node $v_n=\alpha_{D}(v_s)\in V_N$. Therefore, the \textit{General Graph State Distribution Problem} denoted as $\mathscr{P}$ is defined as follows.
\begin{definition}[General Graph State Distribution Problem]
\label{def:definitionOfProblem}
A \textit{General Graph State Distribution Problem} is a triple,
\begin{equation}
\mathscr{P}\triangleq(\mathscr{N},G_S,\alpha_{D}),
\end{equation}
where $\mathscr{N},\,G_S,\,\alpha_D$ is the network triple, graph state, and assignment map as defined before in the text, respectively. 
\end{definition}

A naive distribution protocol to solve $\mathscr{P}$ is to locally create $G_S$ at a central node $n_c\in V_N$ of the network and then establish a Bell pair for each $v_s\in V_S$ between $n_c$ and ${\alpha_D(v_s)}$ for teleporting the desired qubit from the center node to the target node, which is exactly the consideration of the GST algorithm~\cite{fischer2021distributing} developed by Fischer \textit{ et al.}. However, when we reconsider this problem in a distributed manner, another naive strategy naturally arises as establishing a Bell pair for each edge $e=v_{s_1}v_{s_2}$ of the graph state between $\alpha_D(v_{s_1})$ and $\alpha_D(v_{s_2})$, allowing us to consume it for teleporting a CZ gate~\cite{jozsa2006introduction} to build the desired connection. In this case, the connections from a certain vertex can be treated as the 'files' to be distributed, which brings us to the idea of building a P2P protocol. 
\begin{wrapfigure}{r}{0.5\textwidth}
  \vspace{2mm}
  \begin{center}
      \includegraphics[width=.5\textwidth]{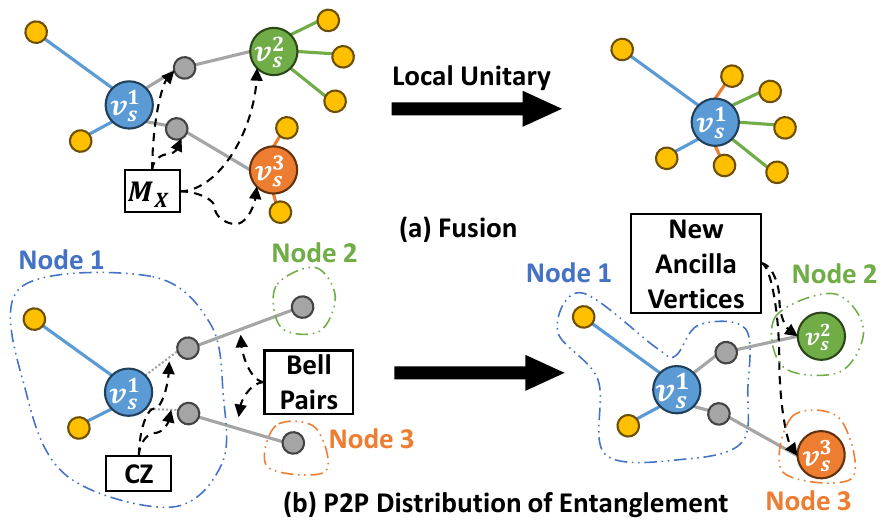}
  \end{center}
  \vspace{-2mm}
  \caption{\label{fig:fusionAP2P}The basic principle of P2P entanglement distribution. (a) Measuring vertices $v_s^2$ and $v_s^3$ and the ancilla qubits (gray circles) in such a configuration followed by local unitary resulting in the merging their connections to other vertices (yellow circles) with $v_s^1$'s. (b) Reversing the fusion operation by connecting the vertex $v_s^1$ to one of the qubits of unused Bell pairs via local CZ gates within node 1 (blue dot-dash line region) resulting in creating ancilla vertices $v_s^2$ and $v_s^3$ in node 2 (green dot-dash line region) and node 3 (orange dot-dash line region) respectively, whose connections can be inherited by a future fusion operation.}
  \Description{Fusion & P2P Distribution}
  \vspace{-2mm}
\end{wrapfigure}

The essential part of Peer-to-Peer file distribution in the classical network is that the node in the P2P network will download the parts of the file not only directly from the original server, but also from the peer nodes that have already received some parts of the file. Meanwhile, the node will also upload the received part of the file, which will help overcome the bottleneck problem in the network. However, this idea cannot be directly applied to quantum communication due to the no-cloning theorem~\cite{wootters1982single}. However, the concept of building a P2P network and not just performing routing at each peer node can be particularly helpful for multiparty entanglement distribution due to the fusion operation as shown in Fig.~\ref{fig:fusionAP2P}(a). When we reverse this procedure, it becomes a P2P scheme in the following sense: the node that established the connection (entanglement) with the destination node ($\alpha_D(v_s)$) of the certain vertex $v_s$ in the graph state can perform a local operation by creating an ancilla and use the network channel to create an ancilla vertex at another node whose connections can be inherited after a future fusion that requires only the LOCC operation, as shown in Fig.~\ref{fig:fusionAP2P}(b). 

In this sense, the entanglement between vertices of the graph state is analogous to the files to be distributed in the classical case, and since all connections to the aforementioned ancilla of a certain vertex can be inherited in the future, those distributed ancilla vertices for the same vertex can be considered as carrying the same file which can be 'copied' and 'distributed' among the network nodes. 

\subsection{Problem Formulation}
We observe that the distribution of the ancilla vertex of a certain vertex will consume an established Bell pair between the source node and the target node, which is the resource of the network. In this way, we can say that this Bell pair is assigned to this certain vertex. And since each channel can generate at most $W_c(e)$ links within a cycle time (shot), this inspires us to define the so-called \textit{One-shot} and \textit{$N$-Shot Bell Pair Strategy} as follows, where the "shot" here means the number of cycles for the network to run. 
\begin{definition}[One-shot \& $N$-Shot Bell Pair Strategy]
\label{def:oneAndNShotBellPairStrategy}
An \textit{One-shot Bell Pair Strategy} $b_1$ is a map $b_1:\vec{E}_N\times V_S\rightarrow \{+1,0,-1\}$ with $0$ representing that the channel is not used by the corresponding vertex and it fulfills the following flow conditions:
\begin{enumerate}
    \item $b_1(\vec{e},v_s)=- b_1(\cev{e},v_s)$,
    \item $\sum_{v_s\in V_S}|b_1(\vec{e},v_s)|\leq W_c(e)$,
\end{enumerate}
where $\vec{E}_N\triangleq\{(e,n_1,n_2)|e=n_1n_2\in E_N\}$ is the oriented edge set of $E_N$ and $0$ represents that the channel is not used by the corresponding vertex. And an \textit{$N$-shot Bell Pair Strategy} is a set denoted as $$\mathscr{b}_N\triangleq\{b_i|i=1,...,N\},$$ where each $b_i$ is an \textit{One-shot Bell Pair Strategy}. 
\end{definition}

Now, we may ask when a $\mathscr{b}_N$ actually solves our problem $\mathscr{P}$, that is, is a \textit{ valid solution}. This question is first addressed by tracking the set of nodes that are possible to hold an ancilla vertex for a certain vertex by the following definition of \textit{Reachable Nodes} and requiring the reachable nodes for the two vertices of any edges to overlap. 
\begin{definition}[Reachable Nodes \& Valid Bell Pair Solution]
$N_{v_s}^{\mathscr{b}_N}\subseteq V_N$ denotes the set of nodes that are reached by $v_s$ under the Bell pair strategy $\mathscr{b}_n$, which is defined in a recursive way by: $n_i\in N_{v_s}^{\mathscr{b}_N} \Leftrightarrow$ either
\begin{enumerate}
    \item $n_{i}=\alpha_D(v_s)\,\mathrm{or}$
    \item $\exists n_j(\neq n_i)\in N_{v_s}^{\mathscr{b}_N},\, b_k\in\mathscr{b}_N, \vec{c}_{ij}=(c_{ij},n_j,n_i)\in \vec{E}_N,\mathrm{s.t.}\,b_k(\vec{c}_{ij},v_s)>0$.
\end{enumerate}
And $\mathscr{b}_n$ is a \textit{Valid Bell Pair Solution} for $\mathscr{P}$ iff
\begin{equation*}
    \forall e=v_{s_1}v_{s_2}\in E_S,\, N_{v_{s_1}}^{\mathscr{b}_N}\cap N_{v_{s_2}}^{\mathscr{b}_N}\neq\emptyset.
\end{equation*}
\end{definition}
Notice that this also gives rise to an equivalent class of valid solutions denoted as $[\mathscr{b}_N]$ counting only the total flow assignment for each vertex, which leads to the following definition of \textit{$N$-Shot Bell Pair Flow Strategy}, 
\begin{definition}[$N$-Shot Bell Pair Flow Strategy]
A \textit{$N$-Shot Bell Pair Flow Strategy} $[\mathscr{b}_N]$ is an equivalent class defined by the following equivalent relation for \textit{$N$-shot Trial Solution} as
\begin{equation*}
    \mathscr{b}_N\sim\mathscr{b}'_N\Leftrightarrow \sum_{b_i\in\mathscr{b}_N}{|b_i(\vec{e},v_s)|}=\sum_{b'_i\in\mathscr{b}'_N}{|b'_i(\vec{e},v_s)|},\,\forall v_s\in V_S,\,\forall \vec{e}\in\vec{E}_N.
\end{equation*}
\end{definition}

However, $\mathscr{b}_N$ itself is not enough to solve the whole problem. The ingredient we are missing is the strategy for using long-term memory to store the necessary auxiliary vertices for the next shot. It is certainly not a trivial problem as we will show later while the developed spacetime network method will allow us to deal with memory usage in a way that is similar to the bell-pair usage and reveal the symmetry of spacetime suggested by the well-known Special Relativity. As shown in Fig.~\ref{fig:spacetime}(a), a quantum transport of a qubit often involves two events in the spacetime manifold, which can be represented as a vector $\vec{\mathscr{e}}=(\Delta x,\Delta t)$ lying inside the light cone in the two-dimensional case. This vector can be decomposed into two vectors, $\vec{\mathscr{e}}_x=(\Delta x,0)$ and $\vec{\mathscr{e}}_t=(0,\Delta t)$, aligned with the $x$-axis and the $t$-axis, respectively. If quantum teleportation is used, we may consider $\vec{\mathscr{e}}_t=(0,\Delta t)$ to represent a procedure of storing a qubit for a time $\Delta t$, while $\vec{\mathscr{e}}_x$ represents an establishment of a Bell pair over a distant $\Delta x$, which decomposes the original communication procedure into two procedures with one of which consumes a Bell pair only and one of which consumes quantum memory only. We remark that the entanglement can be in principle established outside the light cone (within the Entan. Zone) and yields meaningful applications in cryptography~\cite{mayers1998quantum} and collective decision~\cite{ding2024coordinating}. A natural observation is that the difficulty of establishing a pair of bells increases when $\Delta x$ increases, and, typically, the probability suffers from an exponential decay with $\Delta x$ governed by an attenuation factor $\alpha_x$ if no repeaters are used~\cite{RevModPhys.95.045006, li2024generalizedquantumrepeatergraph}. Thus, when a protocol consumes a Bell pair set of $\mathscr{B}=\{(b_i,\Delta x_i)\}$ with $b_i$ here indicating the number of Bell pairs established over $\Delta x_i$, an overall probability due to such consumption can be written as
\begin{equation}
    P_{tot,B}=\prod_i{\mathrm{exp}{[-\alpha_x b_i \Delta x_i]}}=\mathrm{exp}[-\alpha_x \mathscr{C}_{\mathscr{B}}],
\end{equation}
where $\mathscr{C}_{\mathscr{B}}\triangleq\sum_i{b_i\Delta x_i}$ is the \textit{Cumulative Bell Pair Usage}, which reduces to $\#\mathscr{B}$ if all $\Delta x_i$s are the same. Similarly, we can define \textit{Cumulative Memory Usage}, which has been explored both classically~\cite{alwen2015high} and quantumly~\cite{beame2023cumulative} in the previous literature, as follows,
\begin{equation}
    \mathscr{C}_{\mathscr{M}}\triangleq\sum_i{m_i \Delta t_i},
\end{equation}
with $m_i$ denotes the number of qubits that are stored for $\Delta t_i$. We note that cumulative memory usages are closely related to final state fidelity assuming no collective effects and can be considered as the effective quantum volume in some particular quantum systems~\cite{litinski2022active}.

\begin{wrapfigure}{r}{0.5\textwidth}
\vspace{-2mm}
  \begin{center}
      \includegraphics[width=.5\textwidth]{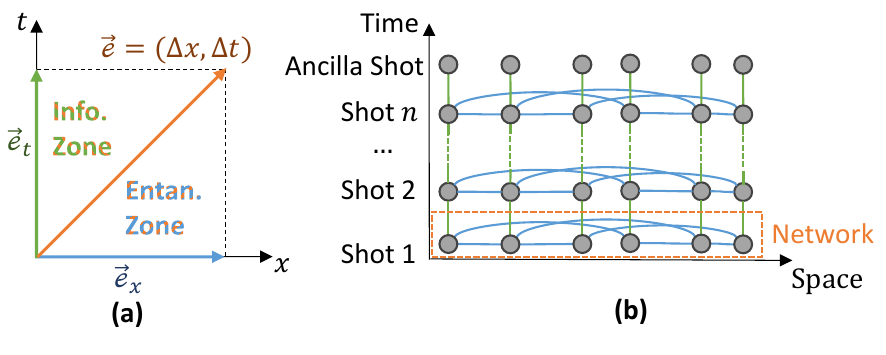}
  \end{center}
  \vspace{-2mm}
  \caption{\label{fig:spacetime}The space-time symmetry. (a) A general communication procedure within the light cone (Info. Zone) $\vec{\mathscr{e}}$ can be decomposed into a memory procedure $\vec{\mathscr{e}}_t$ and a Bell pair distribution procedure $\vec{\mathscr{e}}_x$, the latter of which is outside the light cone (Entan. Zone). (b) The construction of a space-time version of the network (orange dashed box) for $n$ shots.  Nodes are connected by either Bell pair links (blue) or memory links (green). }
  \Description{Space Time Diagram}
  \vspace{-2mm}
\end{wrapfigure}
The symmetry between space and time also inspires us to construct a space-time version of our quantum network as shown in Fig.~\ref{fig:spacetime}(b) to keep track of the Bell pair usage and memory usage simultaneously. Now each node in the space-time version of the quantum network is labeled as $(n_i,k)$ with $n_i\in V_N$ and $k$ indicating the shot coordinate. And the channel set $E_M^{(k)}=\{(n_i,k)(n_i,k+1)|n_i\in V_N\}$ representing the memory usage is also added for each $k$. If an $N$-shot solution is considered, then the node set for the space-time network becomes $V_N\times[N+1]$ and the set of Bell pairs of channels becomes $(E_N\times[N])\cup(E_M\times [N])$, where the ancillary shot is due to the fact that a protocol with $N$ shots will produce snapshots of the network $N+1$. Though this approach has been proposed in the context of the classical ad hoc routing problem within time-varying networks~\cite{zegura2004routing,zhang2019efficient}, our research expands on this by demonstrating its significant effectiveness in quantum entanglement routing. This is due to the ability to form entanglements across the light cone, leading to an ideal space-time symmetry with memory usage. Now, we can define the solution that includes the memory usage as follows with little effort.

\begin{definition}[One-shot Memory Strategy \& $N$-shot Memory Stategy]
\label{def:OANMemoryStrategy}
An \textit{One-shot Memory Strategy} $m_1$ is a map 
\begin{equation*}
    m_1:\vec{E}_M^{(k)}\times V_S\rightarrow\{+1,0,-1\}
\end{equation*}
fulfilling the following flow conditions:
\begin{enumerate}
    \item $m_1(\vec{e},v_s)=- m_1(\cev{e},v_s)$,
    \item $\sum_{v_s\in V_S}|m_1(\vec{e},v_s)|\leq W_m(e)$,
\end{enumerate}
where $\vec{E}_M^{(k)}$ is the oriented set of $E_M^{(k)}$. And the \textit{$N$-shot Memory Stategy} is defined as a collection of \textit{one-shot Memory strategy},
$$\mathscr{m}_N\triangleq\{m_k|k=1,...,N\}.$$
\end{definition}
One notices that such a definition is basically the same as the definition of \textit{Bell Pair strategy} in Def.~\ref{def:oneAndNShotBellPairStrategy}. And for a solution with an explicit memory strategy, denoted as $$\mathscr{S}_N=\{\mathscr{b}_N,\mathscr{m}_N\},$$ we will have a similar definition for when such a solution solves $\mathscr{P}$ as follows.
\begin{definition}[Reachable Nodes under $\mathscr{S}_N$ \& Valid Solution]
\label{def:reachableNodes}
$N_{v_s}^{\mathscr{s}_n}\subseteq V_N\times[N+1]$ denotes the set of nodes that are reached by $v_s$ under the solution $\mathscr{S}_N$, which is defined in a recursive way by: $n_i\in N_{v_s} \Leftrightarrow$ either
\begin{enumerate}
    \item $n_{i}=(\alpha_D(v_s),N+1)\,\mathrm{or}$
    \item $\exists n_j(\neq n_i)\in N_{v_s},\, s_k\in\mathscr{s}_N\in \mathscr{S}_N, \vec{c}_{ij}=(c_{ij},n_j,n_i)\in \vec{E}_N^{(k)}\cup\vec{E}_M^{(k)},\mathrm{s.t.}\,s_k(\vec{c}_{ij},v_s)>0$.
\end{enumerate}
where $E_N^{(k)}=E_N\times\{i\}$ represents the Bell pair channels at $k$th shot. And $\mathscr{S}_n$ is a \textit{Valid Solution} for $\mathscr{P}$ iff
\begin{equation*}
    \forall e=v_{s_1}v_{s_2}\in E_S,\, N_{v_{s_1}}^{\mathscr{s}_N}\cap N_{v_{s_2}}^{\mathscr{s}_N}\neq\emptyset.
\end{equation*}
\end{definition}

\subsection{Optimizing Objectives}
There can be many \textit{Valid Solution}s for $\mathscr{P}$, and comparing them with explicit metrics is important for us to obtain a practical and scalable protocol. The three natural metrics are the \textit{Shot Usage}, \textit{Bell Pairs Consumption} and the \textit{Long-term Memory consumption} as follows.
\begin{definition}[Shot Usage \& Bell Pair Consumption \& Long-term Memory Consumption]
\label{def:shotUsageABellPairConALongtermMemCon} For a \textit{Valid Solution} $\mathscr{s}_N$, its \textit{Shot Usage} ($\#\mathscr{S}$), \textit{Bell Pair Consumption} ($\#\mathscr{B}$) and \textit{Long-term Memory Consumption} ($\#\mathscr{M}$) are defined as follows, respectively. 
\begin{enumerate}
    \item $\#\mathscr{S}\triangleq N$.
    \item $\#\mathscr{B}\triangleq\frac{1}{2}\sum_{i\in[N]}{\sum_{\vec{e}\in \vec{E}_N^{i}}{\sum_{v_s\in V_S}{|b_i(\vec{e},v_s)|}}}$,
    \item $\#\mathscr{M}\triangleq\frac{1}{2}\sum_{i\in[N]}{\sum_{\vec{e}\in\vec{E}_M^{(i)}}{\sum_{v_s\in V_s}{|m_i(\vec{e},v_s)|}}}$.
\end{enumerate}
Here, the $\frac{1}{2}$ factor comes from the fact that each edge is counted twice in the oriented edge set.
\end{definition}
 The first metric is simply the number of cycles required, and it is also the time consumption to execute such a solution. The second is frequently considered in the previous literature~\cite{cuquet2012growth,pirker2018modular,meignant2019distributing,fischer2021distributing,xie2021graph,fan2024optimized}, while the \textit{Long-term Memory Consumption}, as we explained in the previous subsection, is just the cumulative memory~\cite{alwen2015high,beame2023cumulative} and can be considered as a counterpart of the second. We remark that some qubits during execution can be measured instantaneously, while some other qubits will be required to be kept to the next shot. The latter must be kept for at least one cycle time, where the consumption of quantum memory is necessary and the usage is counted by our defined metric.

\subsection{Complexity Analysis}
After all these definitions, our goal is to obtain an algorithm that consumes as few resources as possible. It is obvious that these three resources cannot be optimized simultaneously. Even worse, we present the following theorem which will be proven in the Appendix.~\ref{sec:proofs} that there are no efficient algorithms to minimize any of these three metrics for a general problem $\mathscr{P}$.
\begin{theorem}[Intractability]
Given a problem $\mathscr{P}$, minimizing $\#\mathscr{S}$, $\#\mathscr{B}$ or $\#\mathscr{M}$ are all NP-Hard even if $\alpha_D$ is injective and $G_S$ is connected.
\end{theorem}
The aforementioned theorem is important in the sense that it forces and also rationalizes it to pursue a heuristic algorithm that consumes low resources with relatively low runtime, since it is nearly impossible to have a scalable and practical one that converges to the global minimum unless $P=NP$. However, we get a tight upper bound for $\#\mathscr{S}$ in the case where there is unlimited memory. Here, "tight" means that we construct an instance $\mathscr{P}$ that requires at least the upper bound number of shots to distribute using any algorithm. This result will help us analyze the performance of the proposed algorithms in the worst case and also ensures that the algorithms obtained are within a polynomial approximation (linear to optimize $\#\mathscr{S}$) and thus reliable.
\begin{theorem}[Upper-bound of $\#\mathscr{S}$ Given Unlimited Memory]
Given a problem $\mathscr{P}$, where $G_N$ is connected and $W_m(n_i)=\infty\,,\forall n_i\in V_N$, then $\#\mathscr{S}$ has a tight upper-bound of $\lfloor\frac{|V_S|}{2}\rfloor$.
\end{theorem}
In addition, though the hardness of minimizing $\#\mathscr{M}$ comes from the hardness of minimizing $\#\mathscr{S}$ and one may want to find a $\mathscr{S}_N$ that minimizes $\#\mathscr{M}$ after finding $\mathscr{s}_N$. We can show that there is also no efficient algorithm to do so by the following theorem, which utilizes the developed space-time network method. This not only justifies our early statement that optimizing the memory strategy is a highly nontrivial problem, but also demonstrates the power of the developed spacetime network method. The proofs for these theorems are also included in the Appendix.~\ref{sec:proofs}.
\begin{theorem}[Intractability Given $\mathscr{b}_N$]
Given $\mathscr{b}_N,\,N\geq2$ which is a \textit{Valid Solution} for a problem $\mathscr{P}$, finding $\mathscr{m}_N$ s.t. $\mathscr{S}_N=(\mathscr{b}_N,\mathscr{m}_N)$ that solves $\mathscr{P}$ with minimum $\#\mathscr{M}$ is NP-Hard.
\end{theorem}

\section{ALGORITHMS AND PROTOCOLS}\label{sec:algorithms}
\subsection{Network Model and Routing Protocol}\label{subsec:networkModel}
Our use basically the same network model as the one used by Shi \textit{et al.}~\cite{shi2020concurrent} with slight modification, which establishes an abstraction and corresponds well to the physical quantum devices: 
\begin{enumerate}
    \item The whole network is loosely synchronized in time slots (shot) with duration $T_{\mathrm{shot}}$. We also assume that the information can travel through the whole network within a time slot. 
    \item  The high-quality Bell pair is created in a heralded manner (it can be an abstraction of the physical devices that generate high-quality Bell pairs from multiple tries and entanglement purification, including the quantum repeater\cite{Bennett1996,bernien2013heralded,pirandola2017fundamental}). A channel of width $W_c(e)$ with $e\in E_N$ is capable of creating at most $W_c(e)$ Bell pairs within a time interval, each with a probability of $P_c(e)$.
    \item  All single-qubit operations are assumed to be deterministic, but the only two-qubit operation, $CZ$ operation, is considered probabilistic with probability $P_s$ but in a heralded way~\cite{li2020photon,sung2021realization}.
    \item The short-term memory (qubits are measured almost instantaneously) for each node is assumed to be infinite, while the long-term memory (which can last for at least $T_{\mathrm{Shot}}$) for each node is assumed to be limited by $W_m(n)$ with $n\in V_N$~\cite{heshami2016quantum,hann2019hardware,gu2024hybrid}. The read and write operations are assumed to be deterministic. 
\end{enumerate}
In addition, the distribution of entanglement, as a protocol in the network layer, is also carried out in a cyclic four-phase manner and is based on the previous Bell pair concurrent routing protocol~\cite{shi2020concurrent} as follows:
\begin{enumerate}
    \item \textbf{Synchronization}. The network performs synchronization and uses the information from all the nodes needed to perform local single-qubit operations to complete the required distribution. New requests for the graph state distribution are also sent through the entire network. 
    \item \textbf{Generate Shot Plan}. Use the network topology (which is known to all nodes in advance) to generate the routing routes, assign the qubit, and distribute the Bell pairs between neighbors.
    \item  \textbf{Exchange Link State}. Exchange information on the status of the link among the nodes within $k$ hops of the planned path. $k$ is limited by the duration of short-term memory.
    \item \textbf{Internal Link}. Perform the recovery using the information from the link state and perform the required connection operation using the CZ gate. 
\end{enumerate}
In the second phase, the input to the network to establish Bell pairs is a list of routing routes with index $i$ representing the priority order in a descending way. This is the fundamental strategy for resource allocation of the network: The available resources will be assigned to the path with the higher priority first. 

The discussion on \textit{Cumulative Bell Pair Usage} in the last section inspires us to define the cost of using a Bell pair set $\mathscr{B}=\{(b_i,P_i)\}$ established by the network with $P_i$ representing the probability of success as
\begin{equation}
\label{eq:metric}
    C_{\mathscr{B}}\triangleq -b_i\sum_i \log(P_i).
\end{equation}
 Based on our resource allocation policy, when a width $w$ channel having $o$ width occupied by higher priority paths, the probability $P_c(e|o)$ of such channel $e$ being successful in establishing a new Bell pair for new path can be written as
\begin{equation}
\label{eq:newBellPair}
    P_c(e|o)=P(\#\mathrm{Bell}\geq o+1)=\sum_{i=o+1}^{w}{\binom{n}{i}  P_c^i(e)(1-P_c(e))^{w-i}}
\end{equation}
One can also view it as assigning different probabilities to the channels represented by the parallel edges in a multi-graph. The cost of using each $CZ$ gate can also be represented by $-\log(P_s)$. In this way, the cost metric is defined in an addable manner, and a normal Dijkstra algorithm can be performed to find the optimal path between any two pairs. When all the main paths are found, the additional resources can be used to construct recovery paths. We follow a similar idea as Shi \textit{et al.} but we put forward a new recovery algorithm called the Expected Union Method (EUM), which utilizes the Minimum Cost Maximum Flow~\cite{kiraly2012efficient} algorithm to determine the behavior of the node along each main path, which is further improved to the Spacetime Expected Union Method (ST-EUM) by allowing the path to recovery using the memory links by the spacetime network method. The details of the recovery algorithms are discussed in Appendix.~\ref{sec:rec}.

\subsection{Modified Graph State Transfer Algorithm}
\begin{wrapfigure}{r}{0.5\textwidth}
  \vspace{-2mm}
  \begin{center}
      \includegraphics[width=.5\textwidth]{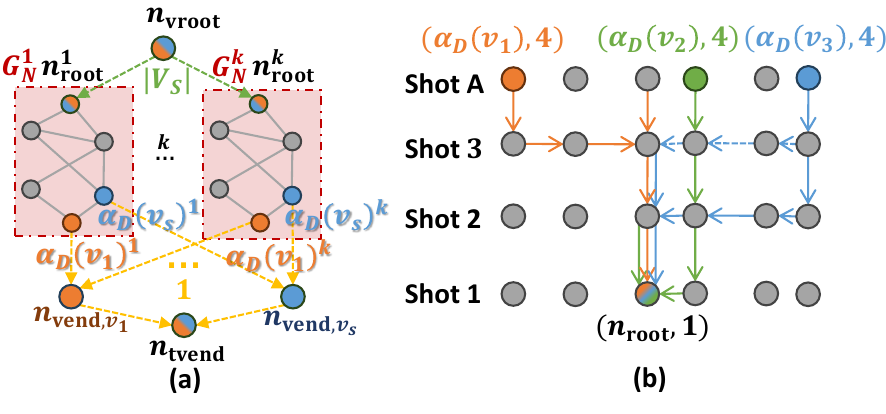}
      
  \end{center}
  \vspace{-4mm}
  \caption{\label{fig:MGST}Illustration for MGST algorithm. (a) The constructed flow graph with the subgraphs inside the red shadow box represents one of the copies of $G_N$ labeled as $G_N^i,\,i=1...k$ together with virtual nodes (larger circles) and virtual arcs (dotted arrows) with the numbers in different colors representing the width of the corresponding channel. (b) An example of $3$-shot solution of MGST. The nodes with multiple colors are the root. The paths in different colors represent the edge-disjoint spacetime path from $\{(\alpha_D(v_s),4\}$ to the root node $(n_{\mathrm{root}},1)$. The blue dotted arcs show the possible recovery possible if the current path fails.}
  \Description{MGST Illustration}
  
\end{wrapfigure}
The GST algorithm proposed by Fischer \textit{et al.}~\cite{fischer2021distributing} is an efficient algorithm but it only gives us an equivalent class $[\mathscr{b}_N]$ that solves $\mathscr{P}$. We need to pick an element $\mathscr{b}_N$ from $[\mathscr{b}_N]$ as well as find a $\mathscr{m}_n$ to form the solution $\mathscr{S}_N=\{\mathscr{b}_N,\mathscr{m}_N\}$ so that our network can execute the solution. One can certainly use the flow decomposition theorem to decompose $[\mathscr{b}_N]$ into a set of paths; however, it is unknown whether such paths can be executed in $N$ shots without conflicting or taking apart specific paths into several shots (which we call it \textit{simply executable}) is unknown. In the worst case, we may have to consume a lot of memory to execute a single path part-by-part across multiple shots, which is definitely unwanted. Hence, how to decompose the flow into a $N$-shot \textit{ simple executable} path set is not a trivial problem. We propose the following novel construction of the flow graph to find the $|V_N|$ edge-disjoint paths from a given root $n_{\mathrm{root}}\in V_N$ executable in $k$ shots:
\begin{enumerate}
    \item Generate the initial flow graph $G_{\mathrm{flow}}$ as a disjoint union of $k$ copies of $G_N$, denoted as $G_N^i,\,i=1...k$. Each node in $G_{\mathrm{flow}}$ will now be labeled with an additional index that determines its shot. Turn $G_{\mathrm{flow}}$ into a directed graph by decomposing each edge into two anti-parallel arcs.
    \item Add a virtual root node $n_{\mathrm{vroot}}$ and add arcs with width $|V_S|$ from $n_{\mathrm{vroot}}$ to all the nodes in $\{n^i_{\mathrm{root}},i=1,...,k\}$ with the superscript being the shot index.
    \item Add $|V_S|$ virtual end nodes $\{n_{\mathrm{vend},v_s},v_s\in V_S\}$ indexed by the vertices in the graph state. For each virtual end  $n_{\mathrm{vend},v_s}$, add an arc with width 1 from each element of $\{\alpha_D(v_s)^i|i=1,...,k\}$ to it.
    \item Add an ultimate virtual end $n_{\mathrm{uvend}}$ and for each virtual end $n_{\mathrm{vend},v_s}$, add an arc with width 1 from it to $n_{\mathrm{uvend}}$.
\end{enumerate}

The illustration for this flow graph is shown in Fig.~\ref{fig:MGST}(a). We remark that such a flow graph is quite similar to the space-time network graph we introduced in the previous section. After such a flow graph is constructed, we run the max-flow algorithm from $n_{\mathrm{vroot}}$ to $n_{\mathrm{uvend}}$. Finding the $|V_N|$ edge-disjoint paths that are \textit{simply executable} from a given root $n_{\mathrm{root}}$ to all the target nodes is possible if and only if the maximum flow is identical to $|V_N|$. Now, applying the flow decomposition theorem to each subgraph $G_N^i$ in the flow graph, we get the desired set of solution paths or $\mathscr{b}_k$.

However, there is one question left to be solved: Will finding such a \textit{ simple executable} path set results in the use of more shots than the flow solution given by GST. This question is answered by the following theorem proven in Appendix.~\ref{sec:proofs}.
\begin{theorem}[Simply Executable Edge-disjoint Path Set]
For all $N$-shot flow solutions for $\mathscr{P}$ given by GST, it is always possible to find a path set that is \textit{simply executable} in $N$-shot and corresponds to a valid solution $\mathscr{s}_n$.
\end{theorem}

\begin{wrapfigure}{r}{0.5\textwidth}
\vspace{-8mm}
\begin{minipage}{0.5\textwidth}
\begin{algorithm}[H]
    \caption{MGST}
    \label{alg:MGST}
    \begin{algorithmic}
       \Function{MGST}{$\mathscr{P}$} 
            \State{$k_\mathrm{min}\gets|V_S|$}
            \State{$n_{\mathrm{root}}\gets Any$}
            \State{$\mathrm{BestCost}\gets\infty$}
            \State{$\mathrm{BestSol}\gets Any$}
            \ForAll{$n_i\in V_N$} 
                \Repeat\,Binary Search on $k$ 
                    \State{$G_{\mathrm{flow}}\gets$ build\_flow\_graph($\mathscr{P}$,$n_i$,k)} 
                    \State{$\mathscr{s}_k\gets$ max\_flow($G_{\mathrm{flow}}$)}
                \Until{Exact $k$ is found}
                \State{Cost$\gets$estimate\_cost($\mathscr{s}_k$)}
                \If{$k<k_\mathrm{min}$} 
                    \State{$k_\mathrm{min},n_{\mathrm{root}},\mathrm{BestSol}\gets k,n_i,\mathscr{b}_k$}
                \ElsIf {$k=k_\mathrm{min}$ \textbf{and} {Cost$<$BestCost} } 
                    \State{$n_{\mathrm{root}},\mathrm{BestSol}\gets n_i,\mathscr{b}_k$}
                \EndIf
                
            \EndFor
            \State{PathSet$\gets$}flow\_decomposition\_theorem($\mathscr{s}_k$)
            \State \Return PathSet
        \EndFunction
    \end{algorithmic}
\end{algorithm}
\end{minipage}
\vspace{-2mm}
\end{wrapfigure}

This modified version of the GST algorithm, MGST, will use the same searching strategy to decide the best root node, as well as the required shot usage in the ideal case. The pseudocode for the complete MGST algorithm is shown in Alg.~\ref{alg:MGST}. The total run time for this algorithm is bounded in the upper order by $O(|V_N|^3|V_S|^3 \max(|E_N|,|V_S|)\log|V_S|)$ using the shortest augmenting path algorithm~\cite{dinic1970algorithm} for max flow.

Now, since the set of solution paths is \textit{ simply executable}, the memory strategy is also quite simple: the root node will store the desired vertex until the path to its destination is expected to be established in this shot. The qubit at the root node will be stored for an additional shot for protection and will be measured for teleportation if the path has been confirmed to be established. If the path fails, a future distribution will be possible and it will not need to restart the whole distribution process. The illustration of this process is shown in Fig.~\ref{fig:MGST}(b). After all, if the desired graph state is established successfully in $N$-shot in the ideal case, the \textit{Memory Consumption} for this protocol is simply calculated as
\begin{equation}
\label{eq:MCMGST}
\#\mathscr{M}_{\mathrm{MGST}}=(N+1)|V_S|-|\{v\in V_s|\alpha_D(v)=n_{\mathrm{root}}\}|,
\end{equation}
where $n_{\mathrm{root}}$ is the root node for distributing such a graph state.

\subsection{Peer-to-Peer Graph State Distribution Algorithm}

 We see that MGST is limited to a special class of solution where edge-disjoint paths are established for each vertex from the assigned node to the center node, or in other words, the channels along these paths are assigned to the specific vertex as shown in Fig.~\ref{fig:MGST}(b). The distribution solution given MGST is completely independent of the topology of the graph state, which can be a merit but can also lead to huge resources being wasted if the graph state is relatively sparse (or has a low degree); an extreme case in point is the star graph. Any star graph can be one-shot distributed in any connected network by assigning all the channels in the oriented spanning tree to the center vertex, or, mathematically,
\begin{equation}
\label{eq:starGraphP2P}
     b_1(\vec{e},v_{\mathrm{center}})=1,\,\vec{e}\in \vec{T}_{\mathrm{spans}\,\{\alpha_D(v_s)|v_s\in V_S\}}
\end{equation}
\begin{wrapfigure}{r}{0.5\textwidth}
\vspace{-2mm}
  \begin{center}
      \includegraphics[width=.5\textwidth]{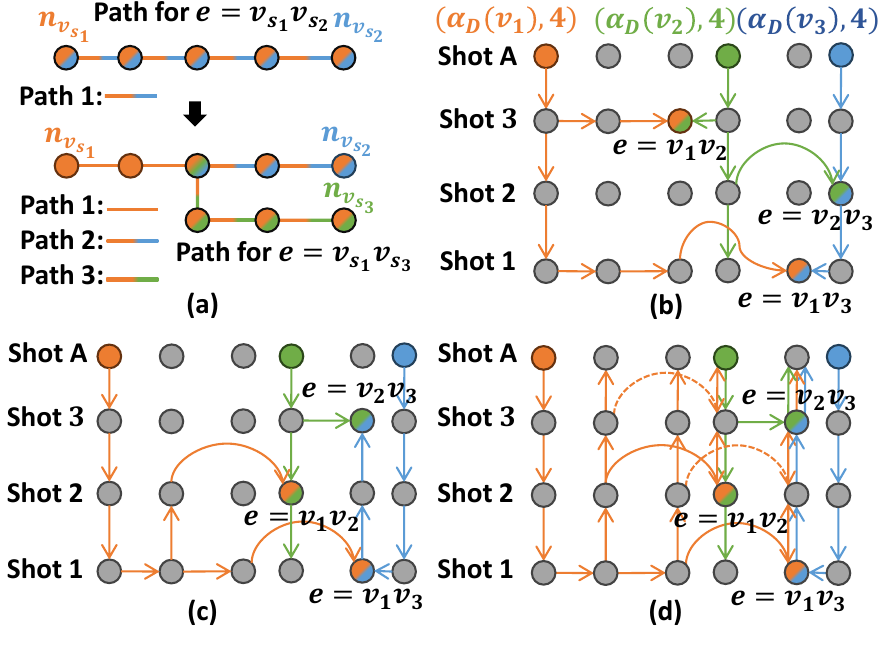}
  \end{center}
  \vspace{-2mm}
  \caption{\label{fig:P2PGSD}Illustration for P2PGSD algorithm. (a) When the path between $n_{v_{s_1}}$ (yellow) and $n_{v_{s_2}}$ (blue) is first found, the assignment of the Bell pairs (yellow \& blue) is not determined. After the new path to $n_{v_{s_3}}$ (green node) is found from an 'unfixed' node along Path 1, part of the Path is fixed and it is split into the new Path 1 and Path 2, while Path 3 is not fixed (yellow \& green). (b) The P2PGSD solution under \textbf{Minimum} memory strategy. The half-and-half colors indicate the node that implements the edge for the corresponding vertices. (c) The P2PGSD solution under \textbf{Standard} memory strategy. (d) The P2PGSD solution under \textbf{Maximum} memory strategy. This might be helpful for later trials (dotted arcs) in case the previous trial fails. }
  \Description{P2PGSD Illustration}
  \vspace{-4mm}
\end{wrapfigure}
In this case, the 'files' to be distributed are the ancilla vertices of the center nodes, which can be regenerated and redistributed in a P2P way at each received node. The Peer-to-Peer Graph State Distribution (P2PGSD) algorithm presents a completely distributed way to distribute the desired graph state across the network utilizing this property. Concretely, the ancilla vertices (the edges) in the graph state are treated as the 'files', and the inversion of merging shown in Fig.~\ref{fig:fusionAP2P}(b) is used to generate and distribute the new ancilla vertices of each vertex.

For the node to obtain the desired ancilla vertex from the nodes that already possess the required vertex, we will have to track which nodes are currently holding the connections from which vertex. This will be done by the so-called Vertex Reaching Map (VRM), which is basically used to track the set of reachable nodes $N^{\mathscr{S}_N}_{v_s}$ for all vertices in Def.~\ref{def:reachableNodes}, or equivalently a map $\mathrm{VRM}:v_s\mapsto N^{\mathscr{S}_N}_{v_s} $. To implement an edge $e=v_{s_1}v_{s_2}$ in the graph state, we will need to distribute the connections from $v_{s_1}$ from the nodes that currently hold the ancilla vertex of $v_{s_1}$ to the nodes that currently hold the one of $v_{s_2}$ or reversely, since the edges are not directed.  This can be completed with a slight modification to the Dijkstra algorithm as follows:
\begin{enumerate}
    \item The inputs to the algorithm are now $\{(n_{v_{s_1},i},c_{s_1,i})\}$ and $\{(n_{v_{s_2},j},c_{s_2,j})\}$. Here, $c_{s_1,i}$ and $c_{s_2,j}$ are the cost of distributing the connections of $v_{s_1}$ and $v_{s_2}$ to the node $n_{v_{s_1},i}$ and $n_{v_{s_2},j}$, respectively.
    \item Initialized the reach table for all $n_{v_{s_1}}$ as reached at cost $c_{s_1,i}$.
    \item Run Dijkstra, when the cost to a node in $\{n_{v_{s_2,j}}\}$ is calculated, the additional cost $c_{s_2,j}$ will be added.
    \item Once the algorithm visits the node with the minimum cost in $\{n_{v_{s_2,j}}\}$, the node and the corresponding path are returned.
\end{enumerate}
We omitted the proof of the correctness for this simple modification version of Dijkstra, while we remark that the correctness of this algorithm relies on the assumption that in the current network,  no path from other $n_{v_{s_1},i}/n_{v_{s_2},j}$ to other $n_{v_{s_1},i'}/n_{v_{s_2},j'}$  with smaller cost exits is available. The modified version of the Dijkstra algorithm has a time complexity of $O(|E_N||V_N|)$ where the additional $|V_N|$ factor comes from the fact that we need to calculate the metric for the whole path at each time to allow a more general metric beyond the addable one, though we can turn our metric into the latter in this work.

\begin{algorithm}[H]
    \caption{P2PGSD}
    \label{alg:P2PGSD}
    \begin{algorithmic}
       \Function{P2PGSD}{$\mathscr{P}$} 
            \State{$k\gets 0$} \Comment{Initialize the shot count}
            \State{$\mathrm{PathSet}\gets \{\}$} \Comment{Initialize PathSet as a dictionary}
            \While{$E_S\neq \emptyset$}
                \State{$\mathscr{N}_{\mathrm{backup}}\gets\mathscr{N}$}
                \State{$k\gets k+1$}
                \State{PathSet$[k]\gets$ []} 
                \State{DegreeList$\gets$ [$v_s$ in descending order of degree]}
                \While{DegreeList is not empty}
                    \State{$v_{sc}\gets$ pop the first element in DegreeList}
                    \State{Neighbors$\gets[v_{sc} \mathrm{'s\, neighbors}]$}
                    \State{sort Neighbors in descending order of degree}
                    \For{$u$ in Neighbors}
                        \State{Path$\gets$Modified\_Dijkstra(VRM[$v_{sc}]$,VRM[$u$]))}
                        \State{append Path to PathSet$[k]$}
                        \State{update $\mathscr{N}_{\mathrm{backup}}$ to remove resources used by Path}
                        \State{update VRM}
                        \State{remove $e=v_{sc}u$ from $E_S$}
                        \State{Sort DegreeList in descending order of degree}
                        
                    \EndFor
                
                \EndWhile
            \EndWhile
            \State \Return PathSet
        \EndFunction
    \end{algorithmic}
\end{algorithm}

We notice that when a path is established between $n_{v_{s_1}}\in N_{v_{s_1}}$  and $n_{v_{s_2}}\in N_{v_{s_2}}$ to implement the edge $e=v_{s_1}v_{s_2}$, the channels (Bell pairs) along this path can be assigned to $v_{s_1}$ or $v_{s_2}$ due to the flexibility of implementing the edge by a local CZ gate on any nodes along this path. We do not need to make the decision immediately when we first find this path; instead, the nodes along this path are added to $\mathrm{VRM}(v_{s_1})$ and $\mathrm{VRM}(v_{s_2})$ with an "unfixed" flag. The decision will be made later when some of the nodes are needed to redistribute the connections from $v_{s_1}$ or $v_{s_2}$, the process of which is shown in Fig.~\ref{fig:P2PGSD}(a). Now, the idea of the P2PGSD algorithm is simply finding the paths for all the edges in the graph state greedily shot by shot while updating the VRM at the same time, the pseudo-code of which is shown in Alg.~\ref{alg:P2PGSD}. The algorithm will not give an explicit memory strategy itself and the actual memory strategy embedded in the adapted protocol has three options as follows:
\begin{enumerate}
    \item \textbf{Minimum}: Only the node in $\{\alpha_D(v_s)\}$ will keep the vertex in place until the next shot. In this case, $\#\mathscr{M}_{\mathrm{P2PGSD\_Min}}=N|V_S|$ if $N$ shot is used.
    \item \textbf{Standard}: All nodes that will be used to redistribute the ancillary vertices in the following shot will be kept to the next shot.
    \item \textbf{Maximum}: All the nodes' ancillary vertices with 'fixed' flag will be kept to the next shot.
\end{enumerate}

The illustration in the space-time network of the solutions given by P2PGSD under different memory strategies is shown in Fig.~\ref{fig:P2PGSD}(b-d). Under the current implementation of P2PGSD, the runtime is upper-bounded by 
\begin{equation}
    O(|V_S||E_S||V_N|\mathrm{max}(|E_N|,|V_S|)),
\end{equation}
which is significantly faster than MGST under all circumstances.

\subsection{Spacetime Peer-to-Peer Graph State Distribution Algorithm}
\begin{wrapfigure}{r}{0.5\textwidth}
  \begin{center}
      \includegraphics[width=.5\textwidth]{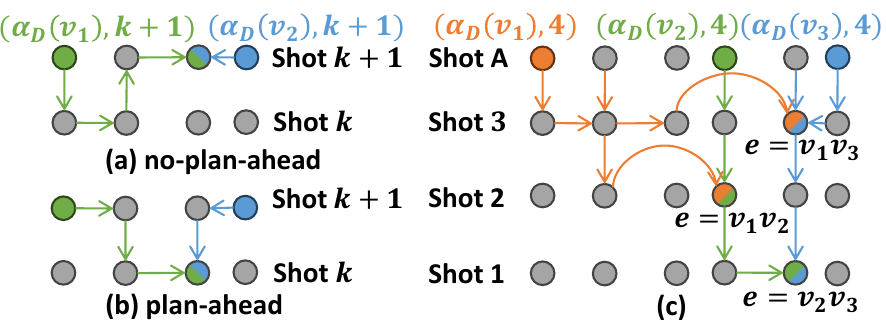}
  \end{center}
  \vspace{-2mm}
  \caption{\label{fig:STP2PGSD}Illustration for STP2PGSD algorithm. (a) A \textit{no-plan-ahead} solution compared to (b) a solution without such limitation, where the parallel arcs in the shot $k$ of the lower plane indicate the violation of such limitation. (b) STP2PGSD directly finds a path in the spacetime network and overcomes the limitation of \textit{no-plan-ahead}. The addition vertical arcs represent a recovery mechanism similar to MGST. }
  \Description{STP2PGSD Illustration}
  \vspace{-4mm}
\end{wrapfigure}
It should be noted that P2PGSD can only converge to a limited class of solutions. One observed from Fig.~\ref{alg:P2PGSD} (c)-(d) that except for the assigned nodes, all other nodes obtained the connections from the specific vertices via routing from the current or the pass shots, which corresponds to the upgoing vertical arcs. This means that these connections need to be preserved using memory for future shots. In addition, if one compares it with the solution given by MGST as shown in Fig.~\ref{fig:MGST}(b), it can be found that all vertical arcs are pointing down, which means the connections are all established in the current or future shots. Physically, the latter means that the vertex of the graph state is created somewhere else and distributed to the assigned node, whereas the former means that the vertex is created at the beginning at the assigned node while its connections are distributed to the reached nodes. In other words, in the P2PGSD solution, the reached nodes are connected to the assigned nodes by routing the previous or current shots. We will call this kind of solution a \textit{no-plan-ahead} solution. A more detailed comparison and illustration of \textit{no-plan-ahead} solution and a solution that is \textit{plan-ahead} is shown in Fig.~\ref{fig:STP2PGSD} (a). We show in Appendix.~\ref{sec:proofs} that this kind of solution can potentially consume more shots than the optimal solution, which calls for an improvement of the previous P2PGSD algorithm.
\begin{wrapfigure}{r}{0.5\textwidth}
\vspace{-6mm}
\begin{minipage}{0.5\textwidth}
\begin{algorithm}[H]
    \caption{STP2PGSD}
    \label{alg:STP2PGSD}
    \begin{algorithmic}
       \Function{STP2PGSD}{$\mathscr{P}$} 
            \Repeat\,Binary Search on $k$
                \State{ST$G_N$$\gets$ Build\_ST\_Network($G_N$,$k$)}
                \State{$\mathscr{P}'\gets$ $\mathscr{P}$ with ST$G_N$}
                \State{TrialPathSet$\gets$P2PGSD($\mathscr{P}')$}
                
            \Until{Exact $k$ is found} \Comment{P2PGSD finishes with one shot}
            \State{PathSet$\gets$TrialPathSet}
            \State \Return PathSet
        \EndFunction
    \end{algorithmic}
\end{algorithm}
\end{minipage}
\vspace{-2mm}
\end{wrapfigure}

The improvement is fairly simple if the spacetime version of the network is utilized. According to Def.~\ref{def:reachableNodes}, the reachable nodes are connected to the assigned node in the spacetime network anyway, either through the links representing memory usage or the links representing the Bell pair usage. In this sense, if we directly run the P2PGSD algorithm on the spacetime network with a determined shot $k$ and replace the assigned node set $\{\alpha_D(v_s)\}$ with $\{(\alpha_D(v_s,k))\}$, this will automatically avoid the limitation of producing \textit{no-plan-ahead} solution of P2PGSD as shown in Fig.~\ref{fig:STP2PGSD}(b), which will also give us an explicit memory strategy. In addition, it is now possible for the pathfinding algorithm to take care of the cost in the form of any linear combination of the \textit{Cumulative Bell Pair Usage} and the \textit{Cumulative Memory Usage} by assigning a proper prefactor to the cost of the memory links and the Bell pair link. If we are more concerned about the shot and memory usage, a larger cost prefactor can be assigned to the memory links. However, such a memory cost will result in preferring the Bell pair links at the higher shots, which does lower the accumulated memory cost when the overall successful probabilities for the Bell pair links are relatively high, but significantly increases the shot usage when the probabilities decrease since the Bell pair links in early shot are not fully explored. This can be resolved by modifying the original cost metric in the following two similar but different ways:
\begin{enumerate}
    \item Assigning a cost factor $m_f$ to change the effective probability of establishing a new Bell pair with the link at $j$-shot as
    \begin{equation}
    \label{eq:metricSTP2PFactor}
        P_c(e|o)\rightarrow (P_c(e|o))^{jm_f},
    \end{equation}
    which results in a factor $jm_f$ when taking the logarithm in the cost.
    \item Observe that if the memory cost is zero, then the paths planned for the previous shot will not occupy the bell link in the current shot only if the bell links have all been generated successfully, adding a factor to the effective probability as
    \begin{equation}
    \label{eq:metricSTP2P}
        P_c(e|o)\rightarrow P_c^{jw}P_c(e|o),
    \end{equation}
    where $j$ is the current shot and $w$ is the width of the channel.
\end{enumerate}
The above two-mentioned methods result in different variants of the ST-P2PGSD algorithm. We note that the cost factor $m_f$ and the memory cost are hyperparameters and can be further optimized to achieve better performance.

The pseudocode for the STP2PGSD algorithm is shown in Alg.~\ref{alg:STP2PGSD} and the runtime of STP2PGSD is upper bound by $O(|V_S|^3|E_S||V_N||E_N|\log|V_S|)$.

\subsection{Adaptive Protocol}

\label{subsec:AdaProtocol}
Due to the fact that the quantum network is probabilistic, it is difficult to establish all connections as planned by the graph state distribution algorithm. It is rather cumbersome and resource-wasting to stick to the distribution plan given by the initial calculation result by the graph state distribution algorithm and wait until all connections for the plan of current to succeed before proceeding to the initial plan for the next shot. Thus, our distribution protocol will run adaptively as shown in Fig.~\ref{fig:adaProtocol}.

\begin{wrapfigure}{r}{.5\textwidth}
\vspace{-2mm}
  \begin{center}
      \includegraphics[width=.5\textwidth]{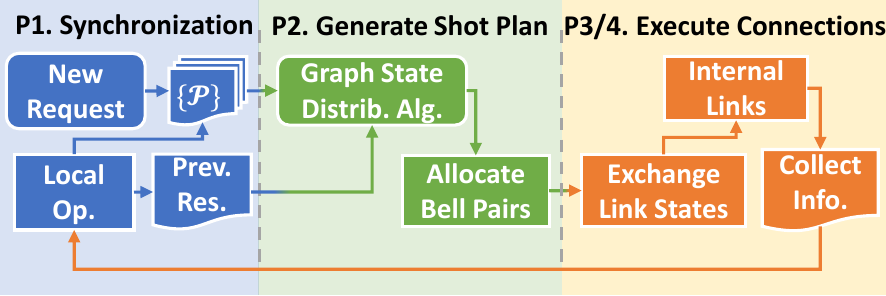}
  \end{center}
  \vspace{-2mm}
  \caption{\label{fig:adaProtocol}Illustration for adaptive graph state distribution protocol. The algorithm is executed in an adaptive manner in order to generate and perform the shot plan that utilizes the network resources efficiently. The arrows indicate the execution order and information transfer relation, and the shaded area indicates which phases the procedure is executed in.}
  \Description{Adaptive Protocol Illustration}
  \vspace{-2mm}
\end{wrapfigure}
At each shot, only the plan for the first shot given by the algorithm will be collected and executed by the network. After that, the execution results for the current shot will be collected in the synchronization phase and the distribution task $\mathscr{P}$ will be modified accordingly. Then the modified distribution task together with the available resources from the last shot will be submitted to the algorithm in the next cycle to perform an adaptive distribution strategy.  This mechanism offers strong flexibility for the graph state distribution in a probabilistic network and also inherits a sense of recovery mechanism using the link resources from the previous shot, which will be beneficial to scaling issues when the probabilities of success for the links in the network are low. However, such an adaptive protocol will require some modifications to the original algorithm to take care of the resources from the previous shot, which is easy and trivial for the first two algorithms but requires careful design for the ST-P2PGSD algorithm. Further details of the adaptive protocol will be left for discussion in the Appendix.~\ref{sec:detailAda}.

\section{Evaluation}\label{sec:evalution}
\subsection{Simulation Setup}\label{subsec:simulationSetup}
The proposed network models as well as the algorithms are implemented on a custom-built simulator in Python with topology generation and common subroutines adapted from the NetworkX package~\cite{SciPyProceedings_11}. As we have mentioned in Sec.~\ref{subsec:networkModel}, the protocols presented can be considered a network protocol that regards general graph state as a basic connection instead of Bell pairs. Thus, our simulations are not built upon other proposed protocols, but rather be implemented from the relatively basic level. The source codes as well as the data for the numerical evaluation can be found in the corresponding open repository~\cite{Huang_P2PGraphStatesDistribution}.

The network topology is generated mainly randomly by the Waxman model~\cite{waxman1988routing} within a unit square area, where the network nodes are placed uniformly unless explicitly mentioned, given the parameters indicating the total number of network nodes and the average degree. Then the success rate of a single usage of a certain channel is also mainly determined by the channel distance with an exponential decay governed by a normalized attenuation factor. Finally, the width for each channel and available long-term memory for each node are sampled from a modified Poisson distribution given the average channel width and average long-term memory width, which ensures each channel of the network has at least one width and the nodes have at least the long-term memory to make the assigned graph state distribution task possible to success.

Multiple types of graph states are being considered including the Erd\H{o}s-R\'enyi graph~\cite{erdds1959random}, star graphs, Pr\"ufer tree graphs~\cite{wang2009optimal}, grid graphs, etc., with some of them being generated randomly when possible. After the graph state topology is generated, the vertices are randomly assigned to the network nodes to implement the assignment map $\alpha_{D}$.
\subsection{Metrology}\label{subsec:metrology}
The simulation is designed to evaluate the time shots, $\#\mathscr{S}$, and cumulative memory consumption, $\#\mathscr{M}$, defined in Sec.~\ref{sec:principle} as well as the computation runtime for distributing various types of graph states (random or fixed) on a random network under the algorithms proposed in this work as well as the one modified from previous work (MGST). We do not focus on Bell pair usage, $\#\mathscr{B}$, since it is not our main optimization goal and it can be shown that it is not compatible with the optimization of the previous two resources and requires a completely different design strategy for the algorithms.

The reference setting configuration for the network topology is a 50-node network lying in a unit area with channels created between two nodes governed by the probability of $p=0.6\exp(-5d)$, where $d$ is the normalized distance relative to the maximum distance between any pair of nodes. And the average channel width is chosen to be around $2$ with an average success rate of around $0.87$ corresponding to a normalized attenuation factor of $0.5$. The long-time memory width is unlimited by default and $P_s$ is set to 1. For the recovery mechanism of the paths, the EUM algorithm is set to be on by default with a maximum hop of 2. We selected EUM rather than the more powerful ST-EUM, since the latter consumes long-term memory of the residue, which makes it difficult to count the actual $\#\mathscr{M}$ used by different protocols and hinders performance evaluation. 

We vary the number of vertices of various types of graphs as well as the nodes from 9 to 200 and evaluate the performance under different average success rates from 0.4 to 1. We also test the algorithms under extremely limited available long-term memory, which is merely enough to distribute the graph state. Each data point is averaged with 1000 independent samples. The performance of the algorithms is compared under various situations mentioned above using the metrics and being understood from our theoretical analysis.

\subsection{Results Analysis}\label{subsec:resAnalysis}
\textbf{The general performance} of protocols based on 3 main algorithms together with two additional variants for P2PGSD and ST-P2PGSD are shown in Fig.~\ref{fig:PerformanceEval1}, where we conclude that the developed P2P series algorithms demonstrate a huge advantage over the center-based MGST algorithm in distributing the widely used tree graphs and grid graphs in the sense of shot usages and cumulative memory usage. The resource consumption difference grows when the number of vertices of the graph state increases, which supports our statement that many resources can be wasted if the graph state topology is not considered for distribution. Furthermore, P2PGSD with a maximum memory strategy consumes more memory than P2PGSD while producing almost the lowest shot consumption, which illustrates the idea of trading memory for faster distribution. It is observed that the memory usage of ST-P2PGSD falls between that of P2PGSD and P2PGSD-Max, and similarly, its shot consumption is intermediate between the two as well. We also investigate the performance of the algorithms under extremely limited memory as shown in Fig.~\ref{fig:PerformanceEval1}(c) where each node only holds the amount of long-term memory that is equal to an average memory of 1 or its assigned vertices to fulfill the minimum requirement to distribute such a graph state (while one of the nodes in the network is chosen to possess $2|V_S|$ memory in order for MGST to be executed). In this case, MGST consumes significantly more resources than the P2PGSD algorithms series and also has huge variance due to the available center node choice. In addition, the performance of various P2PGSD algorithms is drawn to a common line since there is not much room for memory optimization.
\begin{figure*}[t]
  \includegraphics[width=\textwidth]{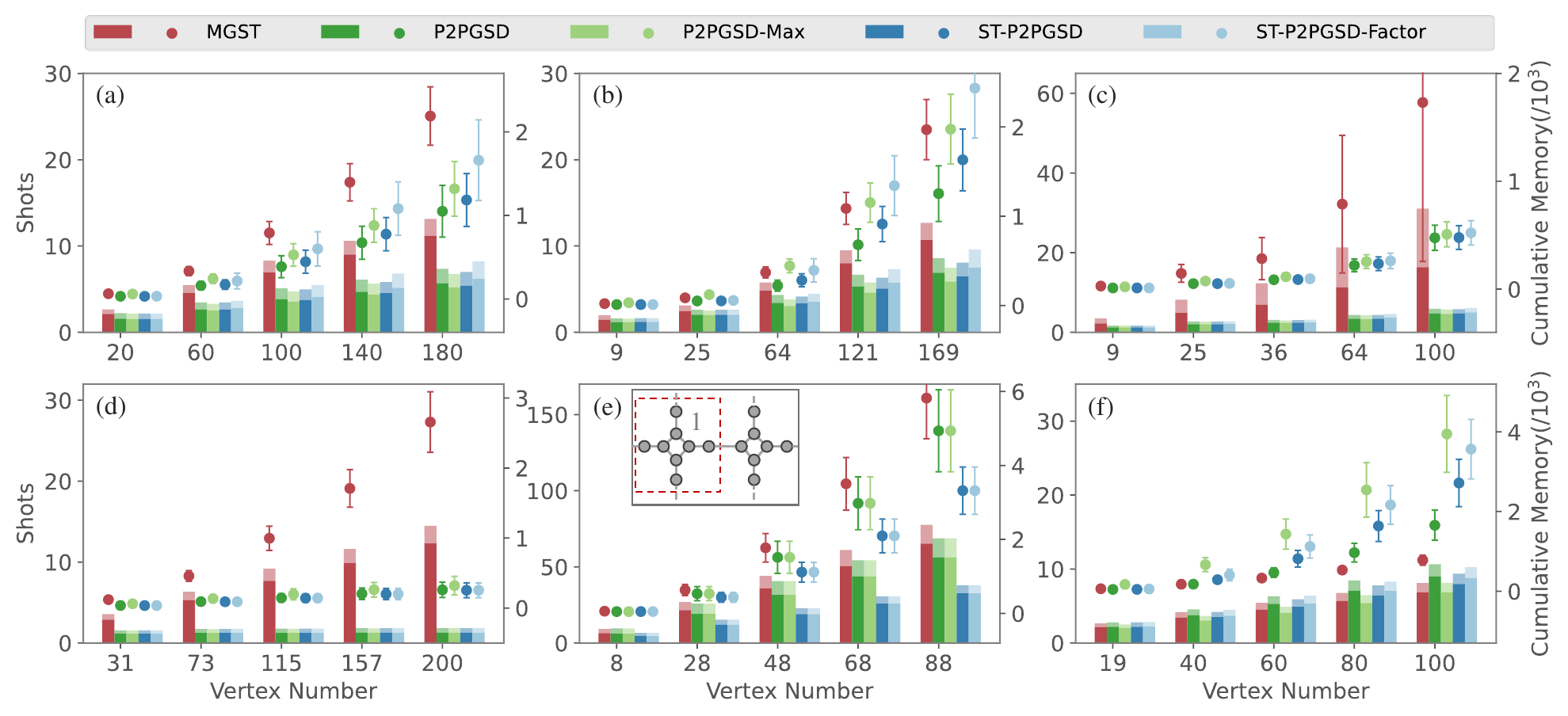}
  \caption{\label{fig:PerformanceEval1}Shots (bar) and cumulative memory (scatter) usage of protocols based on various algorithms versus the number of vertices in certain types of graph state. P2PGSD adapts the standard memory strategy while P2PGSD-Max adapts the maximum memory strategy. The shaped area and the error bar represent the standard deviation of the data. The metric of standard ST-P2PGSD is given by Eq.~(\ref{eq:metricSTP2P}) while the "Factor" version is given by Eq.~(\ref{eq:metricSTP2PFactor}). (a) Random Pr\"ufer tree graphs. (b) Grid graphs. (c) Grid graphs while the performance is evaluated under extremely limited memory. (d) Star graphs. (e) Bell pairs (regular 1) graphs with cell topology network (illustrated in subfigure with width 1 and average channel probability $0.9$). (f) Erd\H{o}s-R\'enyi with average edge probability $0.1$.}
  \vspace{-2mm}
\end{figure*}

\textbf{Advantageous tasks.} Three different main algorithms establish significant advantages with respect to distribution tasks with certain topology characteristics. Fig.~\ref{fig:PerformanceEval1} (d)-(f) shows the typical performance comparison between the algorithms on the presentative tasks. As our initial inspiration for this work, the (ST-)P2PGSD series algorithms demonstrate extremely low resource consumption when the state being distributed is a star graph (GHZ state), as shown in Fig.~\ref{fig:PerformanceEval1} (d). One observes that shot usage as well as cumulative memory consumption hardly grows under the P2PGSD series algorithms, in contrast to MGST. This is a straightforward consequence since distributing such a state only requires the assigned node to be connected in tree form by the successful Bell pair and the resource consumption saturates if the network size is unchanged due to the deterministic local operation. The spacetime version algorithms do not show any advantages in the previous case since these tasks admit a one-shot solution and do not need memory optimization, but one observes from Fig.~\ref{fig:PerformanceEval1} (b) that they consume remarkably fewer resources in terms of both shot usages and memory usages when the network suffers from serious discontinuous bottleneck channels as the case of the so-called "cell topology network" shown in Fig.~\ref{fig:PerformanceEval1} (g). The ST-network method makes it rather easy and convenient to deal with this problem and generate an efficient memory management plan. Finally, the centralized MGST algorithm does show an advantage when the graph state is relatively dense (with a high average degree), as shown in Fig.~\ref{fig:PerformanceEval1} (c), which is understandable due to the fact that it treats vertices instead of edges as 'files' and shall be more efficient when the average number of edges is high and grows quadratically with the number of vertices. However, we remark here that there can be the case that P2PGSD series algorithms offer better solutions even when the graph state is a complete graph, and we offer an example in the Appendix.~\ref{sec:proofs}. 

\begin{wrapfigure}{r}{.5\textwidth}
  \vspace{-2mm}
  \begin{center}
  \includegraphics[width=\linewidth]{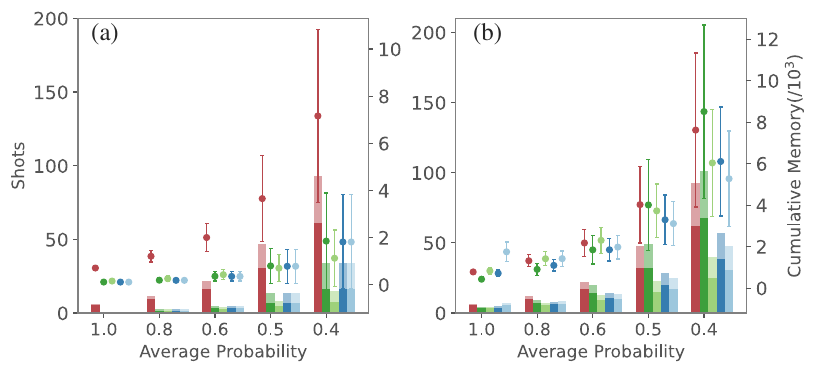}
  \end{center}
  \vspace{-2mm}
  \caption{\label{fig:PerformanceEval2} The performance of protocols based on various algorithms versus the average channel probability is evaluated. The legend is the same as Fig.~\ref{fig:PerformanceEval1} and omitted here. The samples (less than $5\%$ for each data point) that consume more than $200$ shots are discarded to avoid extremeness. (a) The graph state being distributed is a star graph with 115 nodes with random assignment. (b) The graph state being distributed is a $11\times11$-node grid graph. }
  \vspace{-2mm}
\end{wrapfigure}
\textbf{Varying channel probabilities.} The robustness under the probabilistic network of algorithms is investigated under various average channel probabilities assuming the distribution of the star state and the grid state, the results of which are shown in Fig.~\ref{fig:PerformanceEval2}. As we expected, P2PGSD-Max has the best performance in terms of shots when the probabilities decrease, since we are trading memory for speed. In fact, P2PGSD-Max also tends to consume less memory than other algorithms when the probability decreases. 
It is a well-known fact that saving all the unused Bell pairs in a path can result in a better scaling of the expected time for success~\cite{RevModPhys.95.045006}, and this will lead to the threshold behavior in resource consumption, which helps explain the result being observed. Since this is a path-level optimization and can be considered as part of path recovery, we did not implement it in our main distribution algorithms; instead, we offer an analytic solution to the threshold probability for some simple cases in the Appendix~\ref{sec:rec}. In particular, for the grid state distribution case shown in Fig.~\ref{fig:PerformanceEval2} (b), both ST version algorithms perform increasingly better than P2PGSD when probability decreases. It is also worth noting that although the factor version of ST-P2PGSD does not perform well when the probability is high, which may be due to the finite precision effect in calculating the metric, it gains significant advantages in both shot consumption and cumulative memory consumption when the probability decreases. A similar threshold behavior is also observed when comparing its performance with ST-P2PGSD shown in Fig.~\ref{fig:PerformanceEval2}.

\textbf{Scaling and Runtime.} The runtime of producing a solution $\mathscr{S}_N$ for a given instance $\mathscr{P}$ by each algorithm is evaluated on a 3GHz processor and each data point is averaged among 1000 instances, the results of which are shown in Fig.~\ref{fig:PerformanceEval3}.  The estimation of the upper bound of runtime for each algorithm has been given in Sec.~\ref{sec:algorithms}, and it can be concluded that the actual average runtime shows a much better scaling since the actual required shot may not be increased linearly with the vertices number. As we can see from Fig.~\ref{fig:PerformanceEval3} (a), P2PGSD runs the fastest and shows a linear dependence on the number of vertices while the ST version shows a quadratic dependence with MGST in the middle of them. It is worth noticing that there seems to be a turning point in the scaling located roughly around $|V_S|=|V_N|$ for the two ST-algorithms, which might be due to the increasing shots in $\mathscr{S}_n$, where our optimization of the binary search will result in finding the 
\begin{wrapfigure}{r}{.5\textwidth}
  \centering
  \includegraphics[width=\linewidth]{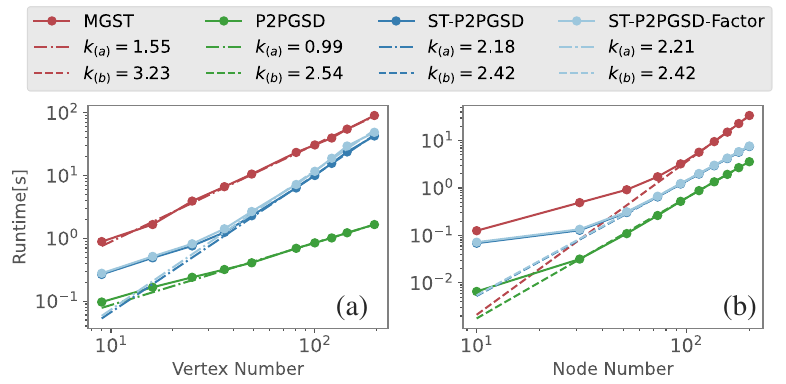}  \caption{\label{fig:PerformanceEval3} Runtime of various algorithms versus scaling of the tasks. The runtime of P2PGSD-Max is not evaluated since it should be basically the same as P2PGSD and the maximum memory strategy is implemented in the adaptive protocol rather than the algorithm itself. The dotted lines are the linear fit of the second half of the data to avoid the influence of the threshold with $k_{(a)}$ and $k_{(b)}$ representing the slope. (a) The performance under various vertex numbers. The graph state being distributed is a grid graph. (b) The performance under various node numbers. The graph state being distributed is a $3\times3$ grid graph. }
  \Description{The runtime performance of algorithms}
  \vspace{-2mm}
\end{wrapfigure}
answer earlier. All P2P-based algorithms show similar scaling with an exponential factor of 2.5 while MGST shows a cubic scaling when increasing the number of nodes in the network, this is somehow expected for the P2PGSD series algorithms due to the nonlinear increasing of the edges. However, MGST performs far better than its upper bound. This may be due to the fact that the maxflow algorithm may have an average run-time scaling close to $\omega(|E_N|)$ instead of its worse upper bound $O(|V_N|^2|E_N|)$ due to the constant maximum flow in our graph state-distributing task.

These results also demonstrate the scalability of our algorithms, with P2PGSD being the fastest among all: for one thing, we show that the dependence is polynomial; for another, the average runtime of getting a solution for all instances involved in the evaluation takes from 0.1 seconds to no more than 2 minutes. Since these algorithms are now implemented in Python without advanced optimization, one also expects an improvement of one to two orders of magnitude when implemented in more efficient programming languages.

\section{DISTRIBUTING ANY QUANTUM COMPUTATION}\label{sec:applitions}
The ability to efficiently distribute any graph state can be essential for distributed quantum computation, since one-way quantum computation, which utilizes the grid state as fundamental and universal resources, is particularly suitable for performing computation among remote qubits\cite{raussendorf2001one}. A naive idea to do so begins with distributing the large grid state among the nodes, with each holding the corresponding qubits of the distributed graph that are required to be measured at that local node. However, in this section, we will show a more efficient way to distribute any computation by teleporting a whole layer of CZ gates.

Before introducing the general distribution strategy, we begin with a motivated example, which is also the initial inspiration for using a P2P scheme in a quantum network. Suppose a server wants to share a particular set of secret quantum states denoted as
\begin{equation}
    |\psi_s\rangle=\alpha|0\rangle+\beta|1\rangle,\,\alpha\beta^*\in i\mathbb{R}
\end{equation}
among $N$ clients (let's say including the server itself for simplicity), where all $N$ clients are required to cooperate to restore this secret but no information about this state can be obtained by any subsets of them, which is also called the $[N,N]$ threshold scheme\cite{cleve1999share}. This can be done by encoding the state into $N$ qubits with $\alpha$ and $\beta$ being the amplitude of the odd parity and even part of the state and then sending each qubit to one client. For example, for a $[3,3]$ scheme, the encoded $3$ qubits state shall be
\begin{equation}
\begin{aligned}
    |\psi_e\rangle&=\frac{\alpha}{2}(|000\rangle+|110\rangle+|011\rangle+|101\rangle)\\&+
    \frac{\beta}{2}(|001\rangle+|010\rangle+|100\rangle+|111\rangle)
\end{aligned}
\end{equation}
This is exactly the center scheme like MGST, while a cleverer way to do it is to use the P2P quantum secret splitting scheme, which can potentially avoid the bottleneck issue in the network: The node that has already received part of the secret and apply the $[2,2]$ scheme as shown in Fig.~\ref{fig:DQC}(a) to split the secret into two and with the new part delivering to the client which has not received part of the secret yet. One can easily show that cascading $[2,2]$ schemes in arbitrary order gives us exactly the same final encoded state as that of the center scheme. We will show later that this is equivalent to using a GHZ state (or a star graph) distributed among the clients as resources to perform such a remote secret-sharing scheme and it is a special and inspiring case of our general distributed quantum computation scheme.

The key requirement of distributing any quantum computation is the ability to teleport any multi-qubit gates since any single-qubit operations can always be done locally and require no entanglement. By selecting a proper universal gate set, say $\{CZ,\,H,\,X,\,Y,\,Z,\,R_Z(\theta)\}$, it is further reduced to the ability to teleport the CZ gate. One can certainly teleport each CZ using a single Bell pair, but one can do better and save more resources by teleporting the whole layer of CZ gates together. To begin with, recall that a graph state $|G_S\rangle$ with underlying graph $G_S=(V_S,E_S)$ can be defined as
\begin{equation}
    |G_S\rangle=\left(\prod_{(i,j)\in E_S}{CZ_{ij}}|+\rangle^{\otimes V_S}\right)\equiv CZ_{E_S}|+\rangle^{\otimes V_S}.
\end{equation}
By the last equality, one can consider that such a graph state is created by applying the entire layer of CZ gates, $CZ_{E_S}$, to the trivial separable tensor state $|+\rangle^{\otimes V_S}$. Now, we want to show that this layer of CZ gates can be transferred to acting on the qubits that are involved in the distributed quantum computation as shown in Fig.~\ref{fig:DQC}(b) using only LOCC operation once the graph state is distributed accordingly. Let $\alpha_D(i)\in V_N$ denote the network node that shall currently hold the qubit $i$ of the graph state and is currently holding the computation qubit labeled as $i''$, where we abuse the notation and use the same letter to represent the one-to-one corresponding of the computation qubit and the graph state qubit. The desired teleportation can be accomplished using the following procedure.
\begin{enumerate}
    \item Distribute the graph state according to $\alpha_D$.
    \item For each $i$, create an auxiliary qubit in the state $|+\rangle$ labeled $i'$ and apply $CZ_{ii'},\,CZ_{i'i''}$ locally.
    \item Measure both the ancilla qubits and the graph state qubits in $X$ basis to get results $m_i'$ and $m_i\in\{+1,-1\}$ .
    \item Apply the local operation $Z_{i''}^{1-m_i\prod_{j\in N_{i}}m_j'}$ with $N_i$ here denotes the neighbors of qubit $i$ in the graph state.
\end{enumerate}
Now, we have the following theorem which will be easily proved by the stabilizer formalism in Appendix.~\ref{sec:proofs}. 
\begin{theorem}[Distributing A Layer of CZ Gates]
Given a graph state $|G_S\rangle$, let $i$ label the graph state qubits and $i''$ label the computation qubits where the same letter means that the two qubits are held at the same node. It is possible to use the LOCC operation to transfer $CZ_{E_S}\equiv\prod_{(i,j)\in E_S}{CZ_{ij}}$ to $\prod_{(i,j)\in E_S}{CZ_{i''j''}}$ using the graph state of the resources $|G_S\rangle.$.
\end{theorem}

We remark that the above CZ gate teleportation strategy makes it possible to gain significant advantages and has the best performance compared to the naive teleportation of a single CZ gate if we can gather as many CZ gates as possible into a single layer. To do so, we shall try to move all the single-qubit gates in our universal gate set sandwiched by two layers of CZ gates out so that the two layers can be merged together. This is easy for $X,\,Y,\,Z,\,R_Z(\theta)$, since commuting them with the CZ gate will give us another single-qubit gate. However, this is not the case for the Hadamard gate. Although one can use single-qubit gate teleportation to move this gate out, this will result in requiring additional qubits in $G_S$ and will not really give us any merits. In fact, this procedure will eventually lead us back to the stage where all the single-qubit gates are performed after a layer of CZ gates, and thus exactly the scheme for one-way computation.

\begin{wrapfigure}{r}{.5\textwidth}
\vspace{-4mm}
  \begin{center}
      \includegraphics[width=.5\textwidth]{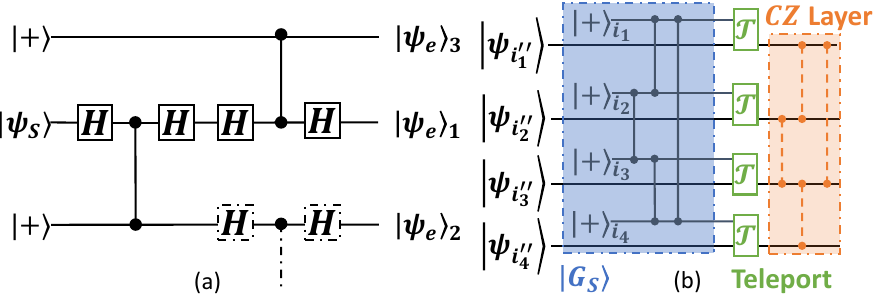}
  \end{center}
  \vspace{-2mm}
  \caption{\label{fig:DQC}Illustration of distributed quantum computation. (a) An inverse CX gate (CZ gates + Hadamard gates) is applied to the quantum secret $|\psi_S\rangle$ to further split it into $|\psi_s\rangle_1$ and $|\psi_s\rangle_2$ in the [2,2] scheme, which can be cascaded to generate more sharing parts. (b) Teleport the whole layer of CZ gates (orange shaded box) using the resources state $|G_S\rangle$ (blue shaded box) and LOCC operation (green boxes). }
  \Description{DQC Illustration}
  \vspace{-4mm}
\end{wrapfigure}
Now, go back to the special case where we want to share a quantum secret among $N$ remote clients. If one thinks of the circuits where the secret is always being split at the center node as shown in Fig.~\ref{fig:DQC}(a), one notices that all the Hadamard gates between the CZ gates are canceled with each other, and thus all CZ gates form a layer that corresponds to a star graph, which justifies our statement earlier. The distribution of a star graph with the (ST-)P2PGSD algorithm also has the exact same pattern as the cascading [2,2] sharing scheme in the sense that the required Bell pairs span a tree that covers all client nodes, which offers an alternative understanding of our previously described P2P quantum secret sharing strategy. We remark that the star graph has also been proven to be the fundamental resource for the quantum sharing of classical secrets as well as the splitting of quantum information~\cite{karlsson1999quantum,hillery1999quantum}, which can also be viewed as a special case of our distributed quantum computation strategy. However, the P2P quantum secret sharing scheme presented here is different from these protocols in the sense that no information can be retrieved for the quantum state if only a proper subset of the clients cooperates together.

\section{RELATED WORK}\label{sec:relate}
Since the invention of quantum teleportation~\cite{bennett1993teleporting},  entanglement has been considered an important resource for large-scale quantum networks. How to efficiently establish reliable pairwise entanglement has been well-studied: Caleffi proposes a novel algorithm to obtain the optimal routing for establishing entanglement between any pair in the network taking account of the probabilistic nature of a specific quantum network model~\cite{caleffi2017optimal}; Pan \textit{et al.} study the routing problem in grid network and relate it to the percolation problem in physics~\cite{pant2019routing}; A comparison between on-demand model and continuous model for quantum network is performed by Chakraborty \textit{et al.}, where suitable algorithms for each model are also proposed~\cite{chakraborty2019distributed}; The first comprehensive study for the protocol of routing multiple Bell pairs on arbitrary network topology is provided by Shi \textit{et al.}~\cite{shi2020concurrent}. Building upon the tremendous work of pairwise entanglement routing, Dahlberg \textit{et al.} design a concrete link layer protocol and study its performance by simulation on real physical parameters~\cite{dahlberg2019link}.

In addition to merely routing the entanglement, inspired by the classical network coding~\cite{li2003linear}, which allows further exploration of the network potential in multicast problems, Kobayashi \textit{et al.} prove that every classical linear network coding scheme can be directly transferred to an effective quantum network coding scheme for the distribution of multiple pairs of bells in two-way communication~\cite{kobayashi2010perfect}. Later, Beaudrap \textit{et al.} show that such linear coding is equivalent to one-way quantum computation~\cite{de2014quantum}, of which the graph states are universal and fundamental resources~\cite{hein2006entanglement}. This naturally leads to the generalization from Bell pairs to graph states as fundamental resources of a quantum network as proposed by Epping \textit{et al.}~\cite{epping2017multi}, who also find a scheme to distribute general two-colorable graph states with network coding constructed from classical linear network codes~\cite{epping2016robust}. 

The distribution of a graph state that matches the network topology was considered by Cuquet \textit{et al.} by locally generating a GHZ state at each node and gluing them together later~\cite{cuquet2012growth}. Pirker \textit{et al.} also generalize their GHZ state-based protocol for stack-like quantum networks to the distribution of any graph state~\cite{pirker2018modular}, resulting in similar protocols to the Edge-Decorated Complete Graph (EDCG) protocol~\cite{meignant2019distributing} proposed by Meignant \textit{et al.} but without concrete metrics on resource consumption. Later, Fischer \textit{et al.} proposed the graph state transfer algorithm (GST) by generating the whole graph state at one central node and then using the max flow algorithm to solve the routing path~\cite{fischer2021distributing}. However, none of these algorithms utilize the topology information of the graph state, and thus many resources can be wasted. As a direct generalization, Xie \textit{et al.} consider a scheme in which multiple center nodes are possible and offer an approximate greedy algorithm to minimize Bell pair usage~\cite{xie2021graph}. It can be shown that this scheme can be further generalized and included in our mathematical definition of the graph state distribution problem we are considering here. We notice that recently, an independent work by Fan \textit{et al.} concerning a similar mathematical model was posted on arXiv, and a linear programming (LP) based algorithm is proposed to find the optimal solution~\cite{fan2024optimized}. However, their algorithms are relatively computationally expensive and only work for specific types of graph states. In addition, the network model and metrics are completely different from ours.

Some works focus on the idea of utilizing network coding for the distribution of graph state, including the one proposed by Epping \textit{el tal.} mentioned above. This is equivalent to locally transforming the desired graph state into a LOCC-equivalent graph state that requires fewer resources to distribute. This problem is extremely complicated, and investigating the orbit of transformation is shown to be $\#$P-complete~\cite{adcock2020mapping}. Lee \textit{et al.} finds an efficient way to decompose the graph state into the 3-star graph resource state while the distribution problem is not really being considered~\cite{lee2023graph}. Koudia explores an approach to using local quantum coding for multipartite graph state transfer, although no universal and explicit protocol is provided~\cite{koudia2023quantum}. Recently Ji \textit{et al.} utilized the simulated annealing (SA) algorithm~\cite{bertsimas1993simulated} to search for the LC-equivalent graph state that requires fewer Bell pairs under a simple distribution scheme of establishing each edge of the graph state by a Bell pair~\cite{ji2024distributing}. We remark that these methods can act as prepossessing and be well combined with our distribution algorithms to give better performance.

\section{CONCLUSION AND OUTLOOK}\label{sec:conclusion}
We present a systematic study on distributing graph state over any quantum network utilizing the peer-to-peer idea from the classical network by first proving the intractability associated with various resource optimization, offering two brand new algorithms and one non-trivially modified algorithm, and then comparing their performance by integrating them into adaptive protocols. It is shown that the P2P-based algorithms show great advantages over the previously proposed centered-based one when the graph state is not dense. A general distributed quantum computation scheme that benefits greatly from our protocols is also proposed with quantum secret sharing as a representative application. We hope that our protocols and algorithms developed here could contribute to the fundamental problem of entanglement generation for future global-scale quantum networks~\cite{huang2024vacuum}. 

In addition, inspired by special relativity, we leverage the space-time network~\cite{zegura2004routing,zhang2019efficient} for quantum entanglement distribution and show how it gives to the counterpart of the Bell pair resource called cumulative memory, which was previously proposed to evaluate amortized complexity~\cite{beame2023cumulative,alwen2015high}. We also demonstrate its power by showing how it can be used to prove the complexity results and construct new algorithms. We believe that this efficient memory management method could be beneficial and act as a convenient tool for future quantum network studies.

We also offer some outlooks for future work here: 1) The P2P series algorithms work particularly well when the graph state is sparse and thus can be combined with other graph state transfer algorithms that reduce the edges of a given graph state~\cite{lee2023graph,ji2024distributing}. 2) Using homology theory to study the topology relation between the network and graph state and find the sufficient conditions for one-shot distribution will be particularly interesting. 3) Incorporating error correction into the protocols may significantly improve performance. 4) Further exploring the spacetime network methods developed here to analyze resource usage for other distributed protocols.

\begin{acks}
We thank Chen Qian, Antonio Barbalace and Allen Zang for their suggestions and acknowledge support from the ARO(W911NF-23-1-0077), ARO MURI (W911NF-21-1-0325), AFOSR MURI (FA9550-19-1-0399, FA9550-21-1-0209, FA9550-23-1-0338), DARPA (HR0011-24-9-0359, HR0011-24-9-0361), NSF (OMA-1936118, ERC-1941583, OMA-2137642, OSI-2326767, CCF-2312755), NTT Research, Packard Foundation (2020-71479), and the Marshall and Arlene Bennett Family Research Program. This material is based upon work supported by the U.S. Department of Energy, Office of Science, National Quantum Information Science Research Centers and Advanced Scientific Computing Research (ASCR) program under contract number DE-AC02-06CH11357 as part of the InterQnet quantum networking project. This work was also completed with resources provided by the University of Chicago’s Research Computing Center and part of this research was performed while the author was visiting the Institute for Mathematical and Statistical Innovation (IMSI), which is supported by the National Science Foundation (Grant No. DMS-1929348).

\end{acks}

\bibliographystyle{ACM-Reference-Format}
\bibliography{p2p}

\appendix

\section{PROOF OF THE THEOREMS AND STATEMENTS}
In this section, we will prove all the theorems and statements in the main text and we will restate the whole theorems for the sake of integrity.
\begin{theorem}[Intractability]
\label{thm:Intractability}
Given a problem $\mathscr{P}$, minimizing $\#\mathscr{S}$, $\#\mathscr{B}$ or $\#\mathscr{M}$ are all NP-Hard even if $\alpha_D$ is injective and $G_S$ is connected.
\end{theorem}
\label{sec:proofs}
\begin{proof}[Proof of Theorem~.\ref{thm:Intractability}]
It suffices to consider the case where the width of the Bell pair channels are all one and there are unlimited long-term memories in each node, which allows us to focus on the property of $G_N$ in the definition of $\mathscr{N}$. The connectivity of $G_S$ as well as the injectivity of $\alpha_D$ actually limit our input to a specific class, and one can easily persuade oneself that if we can show that the theorem stands for the limited case of injective $\alpha_D$ and connected $G_S$, then it also stands for the general case.

First, we begin with the easier one, minimizing $\#\mathscr{B}$. We do this by reducing the Steiner Tree problem to it, which has been known to be one of Karp's 21 NPC problems\cite{karp2010reducibility}.  For any given Steiner Tree instance $\langle G=(V,E),S\subseteq V\rangle$, we construct a $\mathscr{P}$ as follows: let $G_N$ be the same graph as $G$, let $G_S$ be a star graph with $|G_S|=|S|$, finally let $\alpha_D$ be any injective map such that $\alpha_D(V_S)=S$. And we seek for a one-shot solution $b_1$ such that $b_1(\vec{e},v_s)$ may not be zero only for $v_s=v_c$, i.e. all the Bell pairs resources are allocated for the root node. It is obvious that such a solution solves $\mathscr{P}$ if and only if the allocated edges span a tree that covers $\alpha_D(V_S)$, which is exactly the Steiner Tree's problem and completes the proof. We remark that this result suggests the hardness of the optimization for a commonly encountered physical problem of distributing a GHZ state and it has been obtained somewhere else in an implicit way\cite{meignant2019distributing}. 

Now, we proceed to prove the part for minimizing $\#\mathscr{S}$. This time, we prove it by reducing the multi-pair edge-disjoint paths problem to it, which has also been well studied and proved to be NP-Hard\cite{kawarabayashi2011disjoint,even1975complexity}. Given an instance of the edge-disjoint paths problem, 
\begin{equation}
    \langle G=(V,E), \{(n_j^{n_i},n_k^{m_i})\in V\times V|i=1,...,k\}\rangle,
\end{equation}
where $(n_j^{l_i},n_k^{r_i})$ indicates that we shall find a path connecting the node $n_j$ and $n_k$, we construct the $\mathscr{P}$ as follows:
\begin{enumerate}
    \item We first construct $G_N$ from $G$ by adding $2k$ ancilla nodes, with each one connected to an original node via one ancilla edge as shown in Fig.~\ref{fig:ThmS}(a).
    \item $G_S$ is a 1-regular graph with $2k$ vertices with each pair of them representing one of input pair $(n_j^{l_i},n_k^{r_i})$ as shown in Fig.~\ref{fig:ThmS}(b), which also induces a map from the $V_S$ to $V$ by mapping the vertices of the pair to the corresponding nodes denoted by $\beta_D$.
    \item For all vertices $\{v_s|\beta_D(v_s)=n_i, n_i\in V\}$, let $\alpha_D$ be a bijection from the graph state vertices to the set of ancilla nodes s.t. $\alpha_D(\{v_s^{n_i}\})\equiv\{n_i^{(j)}\}$ with $\{n_i^{(j)}\}$ being the set of ancilla nodes connected to the node $n_i$.
\end{enumerate}
It is obvious from the above construction that $\mathscr{P}$ has a one-shot solution if and only if all $k$ edge-disjoint paths are found in the original instance and this directly corresponds to the problem of simultaneously routing Bell pairs for multiple client pairs. One notices that $G_S$ is not connected in this case, while we can make a further but simple modification to $G_N$ and $G_S$ to address this issue. Choose one vertex in $G_S$ for each pair, and connect all pairs via the chosen vertices one by one as shown by the red dotted lines in Fig.~\ref{fig:ThmS}(a). For $G_N$, we perform basically the same modification for the vertices' corresponding ancilla nodes. With the new $\mathscr{P}$ with a connected $G_S$, it is easy to prove that the reduction remains valid while we omit the details of the proof here.

Finally, the minimization of $\#\mathscr{M}$ is closely related to the minimization of $\#\mathscr{S}$: $\#\mathscr{M}$ is lower bounded by $|V_S|$ and the lower bound can be reached if and only if $\mathscr{P}$ can be distributed in one shot. As a result, it is also NP-Hard.
\end{proof}
\begin{figure}[h]
  \centering
  \includegraphics[width=\linewidth]{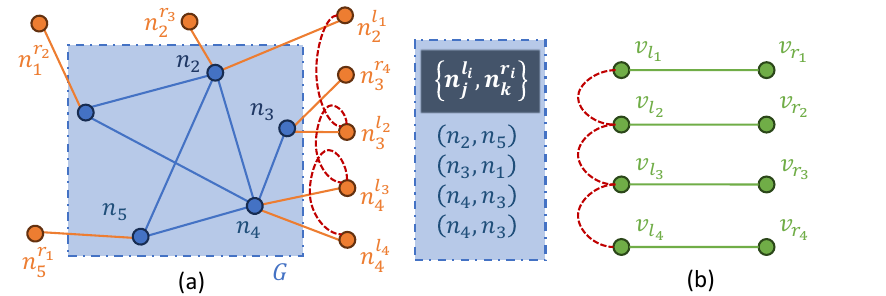}
  \caption{\label{fig:ThmS}Construction of $\mathscr{P}$ from the multi-pair edge-disjoint paths instance. (a) Constructing $G_N$ from $G$ as well as the requested pairs by adding the ancilla nodes as well as ancilla edges(orange).  (b) Construct the  $G_S$ from the given instance. The red dotted line is for further modification to make $G_S$ connected. }
  \Description{ThmS Illustration}
\end{figure}

\begin{theorem}[Upper-bound of $\#\mathscr{S}$ Given Unlimited Memory]
\label{thm:UpperBoundofNS}
Given a problem $\mathscr{P}$, where $G_N$ is connected and $W_m(n_i)=\infty\, ,\forall n_i\in V_N$, then $\#\mathscr{S}$ has a tight upper-bound of $\lfloor\frac{|V_S|}{2}\rfloor$.
\end{theorem}
\begin{proof}[Proof of Theorem~.\ref{thm:UpperBoundofNS}] It is easy to see that MGST guarantees to solve $\mathscr{P}$ within $|V_S|$ shots.  Let $[ \mathscr{b}_{|V_S|}]$ denote its flow solution. It suffices to prove that we can modify it into another solution $[\mathscr{b}'_{|V_S|}]$  with no channel containing flow greater than $\lfloor\frac{V_S}{2}\rfloor$. Find any channel that contains flow greater than $\lfloor\frac{|V_S|}{2}\rfloor$, then use the flow-decomposition theorem to get all the directed paths. Since the paths that utilize the same channel will have the same direction as the theorem guarantees, one can easily show that removing or reversing a subset of paths passing a specific node will not increase the absolute flow values in the channels connecting to it.  There are $f_1>\lfloor\frac{|V_S|}{2}\rfloor$ paths passing this channel to more than $\lfloor \frac{|V_S|}{2}\rfloor$ targets from the roots. All other $f_2=|V_S|-f_1\leq\lfloor \frac{|V_S|}{2}\rfloor$ paths do not use such channel. Then move the root to the downstream endpoint of this channel. Now, we just need to reverse the direction of any $f_2$  paths within $f_1$ paths and remove other $f_1-f_2$ paths to get the new flow solution. Such operation will not create any new channels that carry flow greater than $\lfloor \frac{|V_S|}{2}\rfloor$ as argued above while making the selected channel now consume only $f_2$ flow. This process can be repeated until no channels contain flow greater than $\frac{|V_S|}{2}$.
And one can easily construct an example $\mathscr{P}$ that requires at least $\frac{|V_S|}{2}$ shots to distribute by assuming the graph state consists of merely Bell pairs and the network has a single bottleneck channel with width $1$, which proves that this bound is tight and optimal. We remark that this result is the first tight upper bound, to our best knowledge.

\end{proof}
\begin{theorem}[Intractability Given $\mathscr{b}_n$]
\label{thm:IntractGivenSn}
Given $\mathscr{b}_N,\,N\geq2$ which is a \textit{Valid Solution} for a problem $\mathscr{P}$, finding $\mathscr{m}_N$ s.t. $\mathscr{S}_N=(\mathscr{b}_N,\mathscr{m}_N)$ that solves $\mathscr{P}$ with minimum $\#\mathscr{M}$ is NP-Hard.
\end{theorem}
\begin{proof}[Proof of Theorem~.\ref{thm:IntractGivenSn}] With the help of a space-time network, there is an incredible symmetry between the Bell pair links and the memory links, and so does $\#\mathscr{M}$ and $\#\mathscr{B}$, which inspires us to prove the hardness of this problem by reducing the Steiner tree's problem to it again. For any given Steiner Tree instance $\langle G=(V,E),S\subseteq V\rangle$, we try to construct a $\mathscr{P}$ with $\mathscr{b}_N$.

Again, the graph state $G_S$ is selected to be a star graph with each vertex representing an element in $S\subseteq V$ just like what we do in the $\#\mathscr{B}$ case. And now, WLOG, we let $N=2$ and let $\mathscr{b}_N$ again, only assign the bell links for the center vertex $v_c$ of the graph state. The basic network graph $G_N$ inherits all the nodes from $V$ but with more ancilla nodes and a different edge set. This allows us to make $\alpha_D(V_S)=S\subseteq V_N$ again. To proceed, we directly construct the ST network by first creating 2 nodes and 1 memory link for each node in $V$ then reproducing all the edges from $E$ by mapping it to 4 additional new ancilla nodes and channels according to the following rule: first connecting each endpoint to one ancilla node in the second shot and then connecting the two corresponding ancilla nodes in the second shot through the proper assignment of the values of $b_2$ and $b_1$, respectively, as shown in Fig.~\ref{fig:ThmI}. Now, one can verify that every edge usage in Steiner Tree's instance is mapped bijectively into the usage of two ancilla memory links in the constructed ST network.  The initial Steiner tree problem is minimized if and only if one can find a $\mathscr{m}_2$ with minimum memory usage that connects all the $S$ of the second shot in the constructed ST-network, which completes the proof.
    
\end{proof}
\begin{figure}[h]
  \centering
  \includegraphics[width=\linewidth]{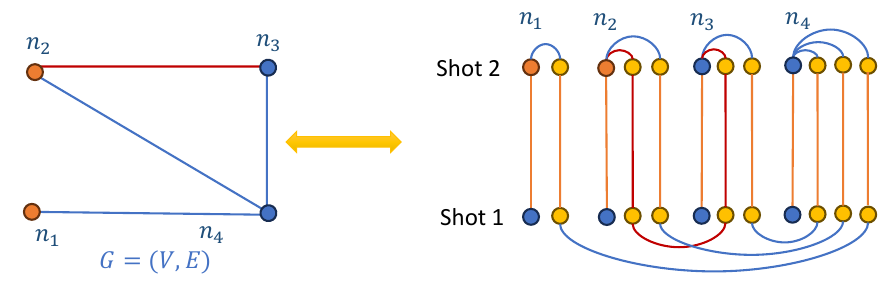}
  \caption{\label{fig:ThmI}Construction of $\mathscr{P}$ and $\mathscr{s}_N$ from Steiner's Tree instance. The left side is the Steiner's Tree instance with orange nodes indicating the vertices from $S$. The right side is the constructed two-shot ST network, where we have omitted the ancilla shot. The blue links are the bell links constructed by making $s_1$ or $s_2$ giving a non-zero flow value to the corresponding channels of $G_N$. The orange links are the memory links and the yellow nodes are the ancilla nodes. The red edges on the left and the red path on the right show the one-to-one corresponding consumption of edges and the memory links. }
  \Description{ThmI Illustration}
\end{figure}
\begin{theorem}[Simply Executable Edge-disjoint Path Set]
\label{thm:SEEDPS}
For all $N$-shot flow solutions for $\mathscr{P}$ given by GST, it is always possible to find a path set that is \textit{simply executable} in $N$-shot and corresponds to a valid solution $\mathscr{S}_n$.
\end{theorem}

\begin{proof}[Proof of Theorem~.\ref{thm:SEEDPS}] 
We prove it by constructing a flow solution for the flow graph shown in Fig.~\ref{fig:MGST} (a) from the flow solution given by GST. Recall that GST seeks a flow solution, denoted as $[\mathscr{b}_N]_{GST}$, for $N$ shots simply by multiplying the capacity of each channel by $N$. We first divide the flow given by $[\mathscr{b}_N]_{GST}$ by $N$ and get a fractional flow that fits perfectly with the capacity requirement of each $G_N^i,\,i=1,...,N$ (we also assign the flow on the virtual edges correspondingly). Notice that such a fractional flow guarantees that the max flow from $n_{\mathrm{vroot}}$ to $n_{\mathrm{tvend}}$ is at least $|V_S|$. Now, using the well-known integral flow theorem\cite{ford1956maximal}, we can convert such a fractional flow solution into the desired integral flow solution for the flow graph, which completes the proof. 
\end{proof}

\begin{theorem}[Distributing A Layer of CZ Gates]
\label{thm:DistribLayerCZ}
Given a graph state $|G_S\rangle$, let $i$ label the graph state qubits and $i''$ label the computation qubits where the same letter means the two qubits are held at the same node. It is possible to use the LOCC operations to transfer $CZ_{E_S}\equiv\prod_{(i,j)\in E_S}{CZ_{ij}}$ into $\prod_{(i,j)\in E_S}{CZ_{i''j''}}$ using the resources graph state $|G_S\rangle.$
\end{theorem}
\begin{proof}[Proof of Theorem~.\ref{thm:DistribLayerCZ}] 
This is a straightforward fact if one views the whole process via the string diagrams\cite{coecke2010quantum} from the categorical quantum computation, while we will prove it using the stabilizer formalism in the Heisenberg picture. The notation used in the proof is the same as that used in the main text and in Fig.~\ref{fig:DQC}. 

Firstly, our goal is to transform the operators for the qubit set $\{i''\}$ in the following way:
\begin{equation}
    Z_{i''}\rightarrow Z_{i''};\,X_{i''}\rightarrow X_{i''}\prod_{j\in N_{i}}{Z_{j''}},
\end{equation}
where $N_{i}$ here denotes the neighbors of vertex $i$ in the graph state.
Notice that the resource graph state as well as the ancilla qubits are stabilized by
\begin{equation}
\begin{aligned}
    K_i&\equiv X_iZ_{N_i}\equiv X_i\prod_{j\in N_i}Z_j\equiv1;\\
    X_{i'}&\equiv 1.
\end{aligned}
\end{equation}
After the local CZ gates, one has the following transformation:
\begin{equation}
    \begin{aligned}
        Z_{i''}&\rightarrow Z_{i''};&\,X_{i''}&\rightarrow X_{i''}Z_{i'};\\
        Z_{i'}&\rightarrow Z_{i'};&X_{i'}\equiv1&\rightarrow X_{i'}Z_{i}Z_{i''}\equiv1;\\
        Z_{i}&\rightarrow Z_{i};\,&X_iZ_{N_i}\equiv1&\rightarrow Z_{i'}X_iZ_{N_i}\equiv1.
    \end{aligned}
\end{equation}
After the measurement is performed, the following operators will be projected into the following values:
\begin{equation}
    X_{i'}=m_{i'};\,X_{i}=m_{i},
\end{equation}
where $m_i$ denotes the corresponding measurement result.
Combining all the equations above, we will have
\begin{equation}
\begin{aligned}
    Z_{i''}&\rightarrow& Z_{i''};&\\
    X_{i''}&\rightarrow& X_{i''}Z_{i'}&=X_{i''}Z_{i'}(Z_{i'}X_iZ_{N_i})(\prod_{j\in N_i} X_{j'}Z_jZ_{j''})\\
    &&&=(m_{i}m_{N_{i'}})X_{i''}Z_{N_{i''}},
\end{aligned}
\end{equation}
where $m_{N_{i}}$ denotes the product of the measurement results of the neighbors of vertex $i$.
The constant factor is then eliminated by local operation $Z_{i''}^{(1-m_{i}m_{N_{i'}})/2}$, which gives us the desired transformation and completes the proof.
\end{proof}

\begin{figure}[h]
  \centering
  \includegraphics[width=\linewidth]{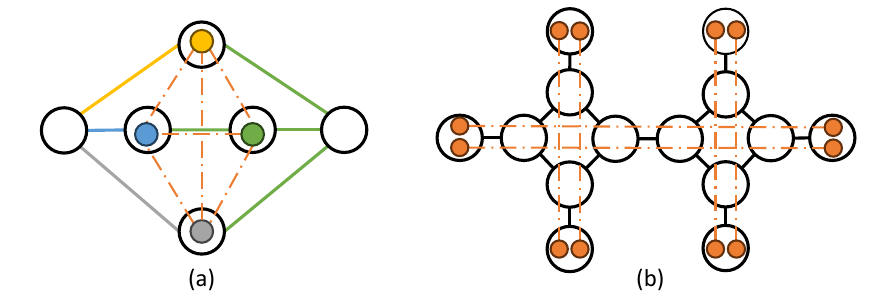}
  \caption{\label{fig:counterExamples} The illustration for the two demonstrative examples. The black circles and lines represent the nodes and width-1 channels of the network while the colorful circles and dashed lines represent the vertices and edges of the graph state.  (a) An example where the graph state is a complete graph. It is impossible for MGST to give a one-shot solution while it is possible for peer-to-peer-based algorithms.  (b) The example where it is impossible to distribute the graph state with a no-plan-ahead solution within two shots. }
\end{figure}

Finally, we offer the two examples mentioned in the main text.
\begin{example}[Counter Example of MGST Being Optimal on Complete Graph]
\label{ex:CEMGST}
The distributed graph state is a complete graph with $4$ vertices with its location indicating the assignment in the network as shown in Fig.~\ref{fig:counterExamples}(a). One can persuade oneself that it is impossible for MGST to give a single shot solution by investigating the cut set while the colors of the channel indicate the channel assignment that offers a one-shot solution.
\end{example}
\begin{example}[Plan-ahead Solution Advantageous Task]
\label{ex:PSAT}
The network topology is the cell topology introduced in the main text with two cells. The graph state being distributed consists of 6 pairs of bell states. One can persuade oneself that it is impossible to find a two-shot solution that is no-plan-ahead due to the bottleneck between two cells which always makes some of the channels idle in the first shot.
\end{example}

\section{RECOVERY ALGORITHM}
\label{sec:rec}
The recovery algorithm is inspired and improved from the recovery algorithm proposed for the Bell pair routing\cite{shi2020concurrent}, which is a path-based algorithm and uses the residue network resource to try to increase the probability of success of the path. The algorithms we proposed here consist of two parts: 
\begin{enumerate}
    \item Recovery using the residue Bell pair resources in the current shot.
    \item Recovery by using the memory links, i.e., saving the successful part of the path to the next shot.
\end{enumerate}
The second part of this algorithm acts as an outlook in the mentioned reference, and we will illustrate the details of the first part and then follow by motivation and the implementation of the second part.

The first part begins by finding the recovery paths based on the given maximum hop and the main paths, which are basically the same as the one in the reference. The algorithm begins by iterating through the hop count and for each hop $k$, iterating through the nodes along the main path to find a recovery path connecting the node and the next $k$-hop node. These recovery paths, $\{P_{\mathrm{rec}}\}$, will be assigned to the main path for recovery usage and the nodes along the recovery paths will switch through the recovery path regardless of the link status if only there are available resources. The key essence is for each node along the main path to determine which two of the qubits inside the node shall perform a switch to maximize the probability of establishing such a path, which is also the improvement of the algorithm we achieve. One may argue that we all perform CZ gates to connect all the qubits, however, this may significantly increase the CZ gates and harm the chance of success if the probability of such a CZ gate is not 1. In addition, to avoid forming a loop, which eliminates the established path section under X measurement, one shall save all the qubits along the paths for at least one shot to determine which qubits to drop, which may also significantly increase the cumulative memory usage. 

The idea of the qubit selection algorithm is simple: The node first updates the probability of channels based on the link state information $L$, which is a bool array of two times the maximum hop $h_{\max}$: the successful links are updated with unit probability, and the unsuccessful links are dropped. Then, the node shall find a path that passes through it and maximize the probability in the residue network formed by the updated channels along the main paths as well as the channels given by the recovery paths with the preassumed probabilities. Such a path can be found using the Minimum Cost Maximum Flow algorithm~\cite{kiraly2012efficient} using a construction similar to MGST treating the node as the source and finding two disjoint paths to the two end nodes. Then the node will choose the two qubits that are along to perform switch such a path if such a path is found. We call this algorithm the Expected Union Method (EUM), which is summarized in Alg. \ref{alg:EUM}. We remark that the key point behind this algorithm is that the events according to each switch choice of that pair are exclusive, and performing the switch will never harm the total chance of success.

\begin{algorithm}
    \caption{EUM}
    \label{alg:EUM}
    \begin{algorithmic}
       \Function{EUM}{$P_{\mathrm{main}},\{P_\mathrm{rec}\},L,v_n$} 
            \State{$\mathscr{N}_{\mathrm{res}}\gets \mathrm{build\_residue\_network}(P_{\mathrm{main}},\{P_\mathrm{rec}\},L,v_n$)}
            \State{$f_{\mathrm{\mathrm{minCostFlow}}}\gets\mathrm{Min\_Cost\_Max\_Flow}(\mathscr{N}_{\mathrm{res}})$}
            \State{$P_{\mathrm{opt}}=\mathrm{flow\_decomposition\_theorem}(f_{\mathrm{\mathrm{minCostFlow}}})$}
            \State \Return{$P_{\mathrm{opt}}$}
        \EndFunction
    \end{algorithmic}
\end{algorithm}

As for recovering using memory, a naive and simple idea is just always to store the qubits as long as there is enough long-term memory. This will indeed lower the shot usage to establish a certain path; however, it will not only cause a blow-up in the usage of cumulative memory but may also waste resources if a high probability link is stored. One observes that there is certainly some threshold behavior in terms of cumulative memory usage, and we shall illustrate it via an intuitive but analytically solvable example here: the task is to establish a Bell pair link along a path with $n$ channels, which has identical probability $p$. We consider a simple storing strategy that treats two shots as one cycle and always saves one qubit if the two switched qubits for each node at the first shot and the qubits at two ends must be saved for two shots. We can calculate the cumulative memory usage to execute such a cycle as $\#\mathscr{C}_{c2}=n+3$, while the expected number of cycles to be executed can be approximated as $T_{c2}\approx\frac{1}{p^n(2-p)^n}$, from which we calculate the expected cumulative memory usage as $\#\mathscr{C}_{\mathrm{tot},\,c2}=\frac{n+3}{p^n(2-p)^n}$. Compared to the no-saving strategy, whose expected cumulative memory usage is $\frac{2}{p^n}$, one immediately deduces a threshold probability as $p_{\mathrm{th}}=2-\left(\frac{n+3}{2}\right)^{\frac{1}{n}}$, in case $n=2$, this corresponds to $p_{th}\approx0.42$, above which one concludes that the no-saving strategy has an advantage. Considering a cycle with $k>2$ shots with an adaptive saving strategy is much more tricky since the neighboring links can be merged into one link to avoid saving the intermediate qubits, although one can easily calculate the expected number of cycles as $T_{ck}\approx\frac{1}{(1-(1-p)^k)^n}$. Even worse, if the memory link is probabilistic in a more general case, the analytical calculation of the expected number of cycles also becomes hardly possible since one can show that this corresponds to a percolation problem in a grid network~\cite{pant2019routing}. 

To recover using the time link, the key idea is that each node along the path decides whether it shall save or discard the qubits. If such recovery needs to be performed, then the current path cannot be completed in the current shot, and the saved resources will be used in the next shot. Thus, we may try to find a path in a two-shot space-time network with the time coordinate for both endpoint nodes lying on the second shot and let the path-finding algorithm decide whether it is necessary and possible to use the resources from the current shot as well as the memory links. The first thing to notice is that there can be a case where a path can not be found for the original version of EUM while it is possible if we consider using the memory links and the channels in the second shot. In this case, if the node does not perform a switch, it may waste the established resources in the current shot. This problem can also be solved by running EUM on the 2-shot space-time network and assigning a small prefactor (which may equal the minimum channel probability of the channel in the first shot) to the probability of each channel in the second shot to encourage the algorithm to use the channels at the current shot as much as possible. Now, the EUM runs on the 2-shot spacetime network, which shall give us the principal path passing the current node. Then we can use the aforementioned pathfinding algorithm to determine the usage of the memory link~\cite{azuma2023minimum}. The two-shot space-time network is constructed as follows:
\begin{enumerate}
    \item The node first runs EUM on the two-shot space-time network to decide its recovered principal path. Here, the probability of the channels in the second shot is given by the formula of Eq.~(\ref{eq:newBellPair}) in the main text with $o$ set to $0$ and multiplied by the small prefactor selected.
    \item Then the node constructs the first shot layer of the spacetime network by adding the first shot channels of its recovered principal path, the cost being the negative logarithm of the probability for each channel along the path to it. 
    \item The node adds a memory link with an adjustable memory cost $\mathrm{memCost}$ (which is default $-\log(0.8)$) for each node along the original main path if there is long-term memory.
    \item The node adds the channels belonging to the original main path to the second shot layer with capacity decided by the following rule: \begin{enumerate}
        \item If the link status of that channel in the current shot is known to the node with a successfully established Bell pair number of $c_s$, and the current main path being the $o_s$-th occupation of this path, then the channel probability is again given by the formula of Eq.~(\ref{eq:newBellPair}) in the main text with $o$ set to $\max(o_s-c_s-1,0)$. The cost is given by taking the negative logarithm. 
        \item If the link status of that channel in the current shot is unknown, then the channel cost is given by the expected channel cost using the formula above over all the possible $c_s$.
        \item If the channel is the channel saved from the previous shot, then it will not have a second shot channel.
    \end{enumerate}
\end{enumerate}
Now, it suffices to check whether the memory link channel has been used in the newly found path. If it is, then the node chooses to save the qubit. We call it the ST-EUM algorithm which is summarized in Alg.~\ref{alg:ST-EUM}.

\begin{algorithm}
    \caption{ST-EUM}
    \label{alg:ST-EUM}
    \begin{algorithmic}
       \Function{ST-EUM}{$P_{\mathrm{main}},\{P_\mathrm{rec}\},L,v_n$} 
            \State{$P_{\mathrm{opt}}=\mathrm{Two\_Shot\_EUM}(P_{\mathrm{main}},\{P_\mathrm{rec}\},L,v_n)$}
            \State{$\mathscr{N}_{\mathrm{res}}\gets \mathrm{build\_ST\_residue\_network}(P_{\mathrm{opt}},P_{\mathrm{main}},L)$}
            \State{$\mathrm{source}\gets \mathrm{Second\_Shot}(\mathscr{N}_{\mathrm{res}},P_{\mathrm
            {main}}[0])$}
            \State{$\mathrm{target}\gets \mathrm{Second\_Shot}(\mathscr{N}_{\mathrm{res}},P_{\mathrm
            {main}}[-1])$}
            \State{$\mathrm{newPath}\gets \mathrm{Dijkstra}(\mathscr{N}_{\mathrm{res}},\mathrm{source,target})$}
            \State{$\mathrm{ifSave}\gets \mathrm{MemLink}(v_n)\,\mathrm{in}\,\mathrm{newPath}$}
            \State \Return{$(P_{\mathrm{opt}},\mathrm{ifSave})$}
        \EndFunction
    \end{algorithmic}
\end{algorithm}
Finally, We remark that EUM is activated by default in the performance evaluation in the main text while ST-EUM is not, since the additional cumulative memory usage by ST-EUM will affect the performance metric of our graph state distribution algorithms.

\section{DETAILS OF ADAPTIVE PROTOCOLS}
\label{sec:detailAda}
The adaptive protocol is designed to make full use of the network information that is synchronized at Phase One of the protocol and to avoid unnecessary waste of resources. Although this might require the algorithms to run multiple times at each run for unfinished distribution tasks, it can significantly improve performance. One might also design a flow algorithm that has a static plan over the whole distribution (which can be easily modified from our algorithms), we argue that the same graph state, unlike the Bell pair resources, is less likely to be requested multiple times; even it is, one can always merge multiple copies of the graph states into a single distribution graph by a disjoint union.  

As mentioned in Sec.~\ref{subsec:AdaProtocol}, the protocol will collect currently established edges of the graph state as well as the resources that are saved using the long-term memory, after which, a new request $\mathscr{P}'$ taking account of this information will be created and submitted to the algorithms, which involves modifying the network topology $\mathscr{N}$ as well as the graph state $G_S$ being distributed. We will illustrate the details of both modifications in this section.

\textbf{Modification of the network.} As we see in Sec.~\ref{sec:rec}, some established Bell pairs will be saved using if there is enough long-term memory, which serves as a recovery over the time axis. Those saved Bell pairs will appear as a one-time channel with a unit probability added to the modified network, $\mathscr{N}$, and submitted to the algorithms. Certainly, some marker will be added to those channels to inform the algorithms that these channels will not be reset in the next shot.

\textbf{Modification of the graph state.} 
The modification is different for the three main algorithms as follows:
\begin{enumerate}
    \item \textbf{MGST}: The modification is fairly easy: the protocol will collect the vertices where the Bell pair has been established successfully between the root and the assigned node. Then those vertices are deleted from the graph state. The long-term corresponding memory will also be free from the center node at the end of Phase One for the next shot. In addition, the root node will be recorded and MGST will not perform a search for the root node in the next shot since the graph state has been generated at the root node.
    \item \textbf{P2PGSD}: The basic idea proceeds in a similar way except that the edges of the graph state instead of vertices are considered as 'files', and those that have been successfully established will be deleted from $G_N$. The key part is to take care of the connections of the vertices at the nodes that are kept to the next shot by the memory strategy. Based on the spirit of peer-to-peer distribution, the node that keeps a certain vertex's connection can be viewed as a newly assigned node for the vertex $v_s$ and we can simply add the node to $\alpha_{D}(v_s)$ which makes the image a subset of $V_N$, which allow the algorithm to construct the initial VRM in the same way. One can delete a certain node from $\alpha_D(v_s)$ if the connection is no longer kept from memory management. Prior to submitting the task, if a certain assigned node does not have enough long-term memory to hold graph state vertex, which could happen only when multiple tasks are processed at the same time, the vertex will be temporally removed from $G_S$. In addition, once the assigned node has enough memory for the vertex, the vertex will be added back to $G_S$ and the memory will be occupied by the vertex until the distribution is finished.
    \item \textbf{ST-P2PGSD}: The protocol for the spacetime version is based on the previous one but way more complicated due to the plan ahead solution. One observation is that the connections established in the previous shot always lie in shot one and never really connected to the assigned node in the ancilla shot when viewed from the spacetime network. The nodes that currently hold the connections for a specific vertex as planned by the algorithm form several connected components using the Bell pairs that have been successfully distributed. The whole distribution process can be viewed as connecting and merging those components together into one using the Bell pairs and CZ gates to get one final component that includes the assigned node and establishes the edge with the components of other vertices. Thus we can create a virtual vertex $v_{v_s}^{(v)}$ to represent each of these components and record the following information:
    \begin{enumerate}
        \item Its source vertex $v_s$ that belongs to the original graph state.
        \item The nodes that currently hold the connection, i.e. the nodes that belong to this component.
        \item The edges that have been established via this component.
        \item The nodes belong to this component and are used to connect this component to the next shot in the spacetime network as planned by the algorithm.
    \end{enumerate}
    Information (b) with time coordinate as shot 1 is defined as the image of $\alpha_D(v_{v_s}^{(v)})$. Then, we also add one virtual edge $(v_s,v_{v_s}^{(v)})$ to $G_S$ to indicate that the component shall be connected to $v_s$ in the future either by a pure memory link (no-plan-ahead) or more Bell pair links. And once two components are connected by the newly established Bell pairs, they will merge. The edges that are successfully established by some current components that are saved in long-term memory will be collected in information (c) and will be temporally removed from $G_S$.

    There have also been many efforts to modify the ST-P2PGSD algorithm to take care of these virtual nodes. The algorithm will search for the path in the ST-Network to implement the virtual edges first based on a specific order determined by information (b) and (c) similar to how we implement the real edges. The VRM for a virtual vertex will not be updated during the execution of the algorithm; instead, once the virtual vertex is connected to the real vertex, its VRM will be merged with the VRM of the real vertex. This method may result in an algorithm that does not give a fixed solution over each shot even if the network has unit channel probability, which may result in some resource waste, while it does converge and work well in practice since the components are being merged together at each shot and expand themselves. Finally, prior to the search for the path in the algorithm, the memory links in the node indicated by the information (d) for the virtual vertices and $\alpha_D(v_s)$ for the real vertices that are reserved for them to prevent failure to obtain a solution in the limited memory case.
\end{enumerate}
If any of the switches (CZ gates) fail, we will adopt a rather pessimistic assumption where the task will need to start from scratch and be reset to the initial task $\mathscr{P}$.

Finally, these protocols allow for multitask processing with priority by feeding the residue network with resources at the current shot used by the previous tasks being removed to the algorithm that processes the next task.
\end{document}